\documentclass[10pt, conference, letterpaper]{IEEEtran}

\usepackage{mathtools,dsfont}
\usepackage{psfrag,bbm,graphicx}
\usepackage{algorithm,algorithmic}
\usepackage{comment,balance}
\usepackage{epstopdf}
\usepackage{epsfig}
\usepackage{multicol}
\usepackage{multirow, cite}
\usepackage{amsfonts}
\usepackage{color}

\newtheorem{corollary}{Corollary}

\newtheorem{lemma}{Lemma}
\newtheorem{remark}{Remark}

\newtheorem{theorem}{Theorem}

\newtheorem{assumption}{Assumption}

\def \OO {\mathrm{O}}
\def \oo {\mathrm{o}}
\def \idle {k}

\begin{document}

\title{Serving Content with Unknown Demand:\\ the High-Dimensional Regime}
\author{
	\IEEEauthorblockN{Sharayu Moharir, Javad Ghaderi, Sujay Sanghavi and Sanjay Shakkottai} \\
}

\maketitle

\begin{abstract}
In this paper we look at content placement in the high-dimensional regime: there are $n$ servers, and $\OO(n)$ distinct types of content. Each server can store and serve $\OO(1)$ types at any given time. Demands for these content types arrive, and have to be served in an online fashion; over time, there are a total of $\OO(n)$ of these demands. We consider the algorithmic task of content placement: determining which types of content should be on which server at any given time, in the setting where the demand statistics (i.e. the relative popularity of each type of content) are not known a-priori, but have to be inferred from the very demands we are trying to satisfy. This is the high-dimensional regime because this scaling (everything being $\OO(n)$) prevents consistent estimation of demand statistics; it models many modern settings where large numbers of users, servers and videos/webpages interact in this way.

We characterize the performance of \textit{any} scheme that separates learning and placement (i.e. which use a portion of the demands to gain some estimate of the demand statistics, and then uses the same for the remaining demands), showing it is order-wise strictly suboptimal. We then study a simple adaptive scheme - which myopically attempts to store the most recently requested content on idle servers - and show it outperforms schemes that separate learning and placement. Our results also generalize to the setting where the demand statistics change with time. Overall, our results demonstrate that separating the estimation of demand, and the subsequent use of the same, is strictly suboptimal.
\end{abstract}

\section{Introduction}
\label{sec:intro}

Ever increasing volumes of multimedia content is now requested and
delivered over the Internet. Content delivery systems (e.g., YouTube
\cite{Youtube}), consisting of a large collection of servers (each
with limited storage/service capability), process and service these
requests. Naturally, the storage and content replication strategy
(i.e., what content should be stored on each of these servers) forms an
important part of the service and storage architecture.\footnote{An earlier version of this work appears in the
	Proceedings of ACM Sigmetrics, Austin, USA, June 2014
	\cite{SGSS14}.}.

Two trends have emerged in such settings of large-scale distributed
content delivery systems. First, there has been a sharp rise in not
just the volume of data, but indeed in \textit{the number
	of content-types} (e.g., number of distinct YouTube videos) that are
delivered to users~\cite{Youtube}. Second, the popularity and demand
for most of this content is \textit{uneven and ephemeral}; in many
cases, a particular content-type (e.g., a specific video clip) becomes
popular for a small interval of time after which the demand
disappears; further a large fraction of the content-types languish in
the shadows with almost no demand~\cite{Gill07,temporal13}.

To understand the effect of these trends, we study a stylized model
for the content placement and delivery in large-scale distributed
content delivery systems. The system consists of $n$ servers, each
with constant storage and service capacities, and $\alpha n$
content-types ($\alpha$ is some constant number). We consider the
scaling where the system size $n$ tends to infinity. The requests for
the content-types arrive dynamically over time and need to be served
in an online manner by the free servers storing the corresponding
contents. The requests that are ``deferred'' (i.e., cannot be
immediately served by a free server with requested content-type) incur
a high cost. To ensure reliability, we assume that there are
alternate server resources (e.g., a central server with large enough
backup storage and service capacity, or additional servers that can be
freed up on-demand) that can serve such deferred requests.

The performance of any content placement strategy crucially depends on
the popularity distribution of the content. Empirical studies in many
services such as YouTube, Peer-to-Peer (P2P) VoD systems, various
large video streaming systems, and web caching,
\cite{Gill07,BC99,YZ06,IRF04,VA02} have shown that access for
different content-types is very inhomogeneous and typically matches
well with power-law (Zipf-like) distributions, i.e., the request rate
for the $i$-th most popular content-type is proportional to $i^{-\beta}$,
for some parameter $\beta>0$.  For the performance analysis, we assume
that the content-types have a popularity that is governed by some
power-law distribution with unknown $\beta$ and further this
distribution changes over time.

Our objective is to provide efficient content placement strategies
that minimize the number of requests deferred. It is natural to expect
that content placement strategies in which more popular content-types
are replicated more will have a good performance. However, there is
still a lot of flexibility in designing such strategies and the extent
of replication of each content-type has to be determined.  Moreover, the
requests arrive dynamically over time and popularities of different
content-types might vary significantly over time; thus the content
placement strategy needs to be online and robust.

The fact that the number of contents is very large and their
popularities are time-varying creates two new challenges that are not
present in traditional queueing systems. First, it is imperative to
{\em measure the performance of content replication strategies over
	the time scale in which changes in popularities occur}. In
particular, the steady-state metrics typically used in queueing
systems are not a right measure of performance in this context.
Second, the number of content-types is enormous and {\em learning the
	popularities of all content-types over the time scale of interest is
	infeasible}. This is in contrast with traditional multi-class
multi-server systems where the number of demand classes does not scale
with the number of servers (low-dimensional setting) and thus learning
the demand rates can be done in a time duration that does not scale
with the system size.

\subsection{Contributions}

The main contributions of our work can be summarized as follows.

\textbf{Modeling Contribution:} We recognize that we are in the high-dimensional regime with unknown demand, that it is fundamentally different from the low-dimensional setting (finite number of content-types) and propose a model that captures this difference.

\textbf{Analytical Contributions:} In Section \ref{subsec:LBSSP}, we show that in this high-dimensional setting where the demand statistics are not known a-priori, the ``\textit{learn-and-optimize}'' approach, i.e., learning the demand statistics from requests and then locally caching content on
servers using the estimated statistics, is strictly
sub-optimal, even when using high-dimensional estimators such as the
Good-Turing estimator \cite{MS00} (Theorem~\ref{thm:converse}). This is in contrast to the
conventional low-dimensional setting where the ``learn-and-optimize'' approach is asymptotically optimal.

In addition, in Section \ref{subsec:MYOPIC}, we study an adaptive content
replication strategy which myopically attempts to cache the most
recently requested content-types on idle servers. Our key result is
that even this simple adaptive strategy strictly outperforms
\textit{any} content placement strategy based on the
``learn-and-optimize'' approach (Theorem~\ref{thm:MYOPIC_static_arrival_rates}). Our results also generalize to the
setting where the demand statistics change with time (Theorems~\ref{thm:converse_2} and \ref{thm:MYOPIC_changing_arrival_rates}).

Overall, our results demonstrate that separating the estimation of
demands and the subsequent use of the estimations to design optimal
content placement policies is deprecated in the high-dimensional
setting.

\subsection{Organization and Basic Notations}
The rest of this paper is organized as follows. We describe our system model and setting in Section \ref{sec:system_model}. The main results are presented in Section \ref{sec:main_results}. Our simulation results are discussed in Section \ref{sec:simulation_results}. Section \ref{proofs} contains the proofs of some of our results. Section \ref{sec:related} gives an overview of related works. We finally end the paper with conclusions. 

Some of the basic notations are as follows. Given two functions $f$ and $g$, we write $f=\OO(g)$ if $\limsup_{n \to \infty} |f(n)/g(n)|<\infty$. $f=\Omega(g)$ if $g=O(f)$. If both $f=\OO(g)$ and $f=\Omega(g)$, then $f=\Theta(g)$. Similarly, $f=\oo(g)$ if $\limsup_{n\to \infty} |f(n)/g(n)|=0$, and $f=\omega(g)$ if $g=\oo(f)$. The term $w.h.p.$ means with high probability as $n \to \infty$.  

\section{Setting and Model}
\label{sec:system_model}
In this section, we consider a stylized model for large scale
distributed content systems that captures two emerging trends, namely,
a large number of content types, and uneven and time-varying demands.

\subsection{Server and Storage Model}

The system consists of $n$ front-end servers, each of which can hold one content piece, and serve one user, at any time. In addition, there is a back-end server that stores the entire
catalog of $m$ content-types (one copy of each content-type, e.g., a
copy of each YouTube video). The contents can be copied from the
back-end server and placed on the front-end servers. 

Since we are interested in the scaling performance, as $n,m \to \infty$, for clarity we assume that there are $n$ servers and each server can store $1$ content and can serve $1$
request at any time. If instead of one content, each front-end server can store at most $d>1$ content pieces ($d$ is a constant) and serve at most $d$ requests at each time, the performance can be bounded from above by the performance of a system with $d n$ servers with a storage of $1$ each, and from below by that of another system with $n$ servers with a storage of $1$ each. Thus asymptotically in a scaling-sense, the system is still equivalent to a system of $n$ servers where each server can store $1$ content and can serve $1$ content request at any time.


\subsection{Service Model}
When a request for a content arrives, it is routed to an idle
(front-end) server which has the corresponding content-type stored on
it, if possible. We assume that the service time of each request is
exponentially distributed with mean 1. The requests have to be served
in an online manner; further service is non-preemptive, i.e., once a
request is assigned to a server, its service cannot be interrupted and
also cannot be re-routed to another server. Requests that cannot be
served (no free server with requested content-type) incur a high cost
(e.g., need to be served by the back-end server, or content needs to
be fetched from the back-end server and loaded on to a new server). As
discussed before, we refer to such requests as deferred requests. The
goal is to design content placement policies such that the number of
requests deferred is minimized.

\subsection{Content Request Model}
There are $m$ content-types (e.g., $m$ distinct YouTube videos). We
consider the setting where the number of content-types $m$ is very
large and scales linearly with the system size $n$, i.e., $m = \alpha
n$ for some constant $\alpha>1$. We assume that requests for
each content arrive according to a Poisson process and request rates
(popularities) follow a Zipf distribution. Formally, we make the
following assumptions on the arrival process.
\begin{assumption}(Arrival and Content Request Process)
	\label{ass:zipf}
	\begin{itemize}
		\item[-] The arrival process for each content-type $i$ is a Poisson
		process with rate $\lambda_i$.
		
		\item[-] The load on the system at any time is $\bar{\lambda} < 1$, where
		$
		\bar{\lambda} = \dfrac{\sum_{i=1}^m \lambda_i}{n}.
		$
		
		\item[-] Without loss of generality, content-types are indexed in the
		order of popularity. The request rate for content-type $i$ is
		$\lambda_i = n \bar{\lambda}p_i$ where $p_i \propto i^{-\beta}$ for
		some $\beta> 0$. This is the Zipf distribution with parameter
		$\beta$.

	\end{itemize}
\end{assumption}
We have used the Zipf distribution to model the popularity
distribution of various contents because empirical studies in many
content delivery systems have shown that the distribution of
popularities matches well with such distributions, see e.g.,
\cite{Gill07}, \cite{BC99}, \cite{YZ06}, \cite{IRF04}, \cite{VA02}.

\subsection{Time Scales of Change in Arrival Process}
\label{subsec:change_model}
A key trend discussed earlier is the time-varying nature of
popularities in content delivery systems \cite{Gill07,temporal13}. For
example, the empirical study in \cite{Gill07} (based on 25 millions
transactions on YouTube) shows that daily top 100 list of videos
frequently changes.
To understand the effect of this trend on the performance of content
placement strategies, we consider the following two change models.\\

\noindent \textbf{Block Change Model:} In this model, we assume that
the popularity of various content-types remains constant for some
duration of time $T(n)$, and then changes to some other arbitrarily
chosen distribution that satisfies Assumption \ref{ass:zipf}. Thus
$T(n)$ reflects the time-scale over which changes in popularities
occur. Under this model, we characterize the performance of content
placement strategies over such a time-scale $T(n)$. \\

\noindent \textbf{Continuous Change Model:} Under this model, we
assume that each content-type has a Poisson clock at some constant
rate $\nu>0$. Whenever the clock of content-type $i$ ticks,
content-type $i$ exchanges its popularity with some other content-type
$j$, chosen uniformly at random. Note that the average time over which
the popularity distribution ``completely'' changes is comparable to
that of the Block Change Model; however, here the change occurs incrementally
and continuously. Note that this model ensures that the
content-type popularity always has the Zipf distribution.
Under this model, we characterize the performance of content placement
strategies over constant intervals of time.

\section{Main Results and Discussion}
\label{sec:main_results}

In this section, we state and discuss our main results. The
proofs are provided in Section \ref{proofs}.


\subsection{Separating Learning from Content Placement}
\label{subsec:LBSSP}
In this section, we analyze the performance of storage policies which
separate the task of learning and that of content placement as
follows. Consider time intervals of length $T(n)$. The operation of
the policy in each time interval is divided into two phases:\\

\noindent \textbf{Phase 1. Learning:} Over this interval of time, use
the demands from the arrivals (see Figure
\ref{fig:learning_based_storage}) to estimate
the content-type popularity statistics. \\\\
\noindent \textbf{Phase 2. Storage:} Using the estimated popularity of
various content-types, determine which content-types are to be replicated
and stored on each server. The storage is fixed for the remaining time
interval. The content-types not requested even once in the learning
phase are treated equally in the storage phase. In other words, the
popularity of all {\em unseen} content-types in the learning phase is
assumed to be the same.\\

\begin{figure}[h]
	\begin{center}
		\includegraphics[scale=0.4]{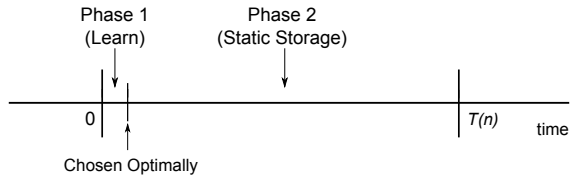}
		\caption{Learning-Based Static Storage Policies -- \sl The interval
			$T(n)$ is split into the Learning and Storage phases. The length of time spent in the Learning phase can be chosen optimally using the knowledge of the value of $T(n)$ and the Zipf parameter $\beta$. \label{fig:learning_based_storage}}
	\end{center}
\end{figure}

Further, we allow the interval of time for the Learning phase potentially
to be chosen optimally using knowledge of $T(n)$ (the interval over
which statistics remain stationary) and $\beta$ (the Zipf parameter
for content-types popularity).

%

This is a natural class of policies to consider because it is obvious
that popular content-types should be stored on more servers than the
less popular content-types. Therefore, knowing the arrival rates can
help in the design of better storage policies. Moreover, for the
content-types which are not seen in the learning phase, the storage
policy has no information about their relative popularity. It is
therefore natural to treat them as if they are equally popular.

The replication and storage in Phase~2 (Storage) can be performed by
\textit{any} static policy that relies on the knowledge (estimate) of
arrival rates, e.g., the proportional placement policy \cite{LLM12}
where the number of copies of each content-type is proportional
to its arrival rate, or the storage policy of \cite{LLM13} which was
shown to be approximately optimal in the steady state.

We now analyze the performance of learning-based static storage
policies under the Block Change Model defined in Section
\ref{subsec:change_model} where the statistics remain invariant over
the time intervals of length $T(n)$. The performance metric of
interest is the number of requests deferred by any policy belonging to
class of learning-based static storage policies in the interval of
interest. We assume that at the beginning of this interval, the
storage policy has no information about the relative popularity of
various content-types. Therefore, we start with an initial loading
where each content-type is placed on exactly one server. This loading is
not changed during Phase 1 (the learning phase) at the end of which,
the content-type on idle servers is changed as per the new storage
policy. As mentioned before, this storage is not changed for the
remaining duration in the interval of interest.



Theorem 1 in \cite{SGSS14} provides a lower bound on the number of requests
deferred by any learning-based static storage policy for the Block Change Model for the Zipf distribution with parameter $\beta > 2$. The following theorem provides a stronger bound on the performance of all learning based policies to extend this result for $\beta > 1$. This includes the case where $\beta  = 1.2$, known to be a good fit for Video on Demand (VoD) systems \cite{fricker2012impact}.

\begin{theorem}
	\label{thm:converse}
	Under Assumption \ref{ass:zipf} and the Block Change Model defined in
	Section~\ref{subsec:change_model}, for $\beta > 1$, if $T(n) =
	\Omega(1)$, the expected number
	of requests deferred by any learning-based static storage policy is
	$\Omega\big(\min\{(nT(n))^{\frac{1}{2-1/\beta}},n\}\big)$.
\end{theorem}

We therefore conclude that even if the division of the interval of
interest into Phase~1 (Learning) and Phase 2 (Storage) is done in the
optimal manner, no learning-based static storage policy can defer
fewer than $\Omega\big((nT(n))^{\frac{1}{2-1/\beta}}\big)$ jobs in the interval of
interest. Therefore, Theorem \ref{thm:converse} provides a fundamental
lower bound on the number of jobs deferred by any policy which
separates learning and storage. It is worth pointing out that this
result holds even when the time-scale of change in statistics is quite
slow. Thus, even when $T(n)$, the time-scale over which statistics
remains invariant, goes to infinity and the time duration of the two phases
(Learning, Storage) is chosen optimally based on $\beta$, $T(n)$, $\Omega\big(\min\{(nT(n))^{\frac{1}{2-1/\beta}},n\}\big)$ requests are still deferred.

The next theorem provides a lower bound on the number of requests deferred by any learning-based static storage policy for the Continuous Change Model.  As before, we assume that at the beginning of this interval, the storage policy has no information about content popularity and therefore, we start with an initial loading where each content-type is placed on exactly one server.

\begin{theorem}
	\label{thm:converse_2}
	Under Assumption \ref{ass:zipf} and the Continuous Change Model defined in
	Section~\ref{subsec:change_model}, for $\beta > 1$, if $T(n) =
	\Omega(1)$, the expected number
	of requests deferred by any learning-based static storage policy is
	$\Omega\big(\min\{(nT(n))^{\frac{1}{2-1/\beta}},n\}\big)$.
\end{theorem}


Next, we explore \textit{adaptive storage policies} which perform the
task of learning and storage simultaneously.


\subsection{Myopic Joint Learning and Placement}
\label{subsec:MYOPIC}
We next study a natural adaptive storage policy called MYOPIC. In an
adaptive storage policy, depending on the requests that arrive and
depart, the content-type stored on a server can be changed when the server
is idle while other servers of the system might be busy serving
requests. Therefore, adaptive policies perform the tasks of learning
and placement jointly. Many variants of such adaptive policies have
been studied for decades in the context of cache management (e.g. LRU,
LRU-MIN \cite{LRU-MIN95}).


Let $C_i$ refer to the $i^{th}$ content-type, $1 \leq i \leq m$. The
MYOPIC policy works as follows: When a request for content-type $C_i$
arrives, it is assigned to a server if possible, or deferred
otherwise. Recall that a deferred request is a request for which on
arrival, no currently idle server can serve it and thus its service
invokes a backup mechanism such as a back-end server which can serve
it at a high cost. After the assigment/defer decision is made, if
there are no currently idle servers with content-type $C_i$, MYOPIC
replaces the content-type of one of the idle servers with $C_i$. This idle
server is chosen as follows:
\begin{itemize}
	\item[-] If there is a content-type $C_j$ stored on more than one currently
	idle server, the content-type of one of those servers is replaced with
	$C_i$,
	\item[-] Else, place $C_i$ on that currently idle server whose content-type
	has been requested least recently among the content-types on the
	currently idle servers.
\end{itemize}
\noindent For a formal definition of MYOPIC, refer to Figure \ref{fig:MYOPIC}.\\

\begin{figure}[h]
	\hrule
	\vspace{0.1in}
	\begin{algorithmic}[1]
		\STATE On arrival (request for $C_i$) \textbf{do},
		\STATE Allocate request to an idle server if possible.
		\IF {no other idle server has a copy of $C_i$,}
		\IF {$\exists j$: $C_j$ stored on $> 1$ idle servers,}
		\STATE replace $C_j$ with $C_i$ on any one of them.
		\ELSE
		\STATE find $C_j$: least recently requested on idle servers,\\ replace $C_j$ with $C_i$.
		\ENDIF
		\ENDIF
	\end{algorithmic}
	\vspace{0.1in}
	\hrule
	\caption{MYOPIC -- \sl An adaptive storage policy which changes the content stored on idle servers in a greedy manner to ensure that recently requested content pieces are available on idle servers.}
	\label{fig:MYOPIC}
\end{figure}

\begin{remark} Some key properties of MYOPIC are:
	\begin{enumerate}
		\item The content-types on servers can be potentially changed only when
		there is an arrival.
		
		\item The content-type of at most one idle server is changed after each
		arrival. However, for many popular content-types, it is likely that
		there is already an idle server with the content-type, in which case
		there is no content-type change.
		
		\item To implement MYOPIC, the system needs to keep track of the time
		at which the recent most request of each content-type was made.
		
	\end{enumerate}
\end{remark}

The following theorem provides an upper bound on the number of
requests deferred by MYOPIC for the Block Change Model defined in Section
\ref{subsec:change_model}.
\begin{theorem}
	\label{thm:MYOPIC_static_arrival_rates}
	Under Assumption \ref{ass:zipf} and the Block Change Model defined in
	Section~\ref{subsec:change_model}, over any time interval $T(n)$ such
	that $T(n) = \oo(n^{\beta-1})$, the number of requests deferred by
	MYOPIC is $\OO((nT(n))^{1/\beta})$ w.h.p.
\end{theorem}
We now compare this upper bound with the lower bound on the number of
requests deferred by any learning-based static storage policy obtained
in Theorem ~\ref{thm:converse}. 

\begin{corollary}
	\label{cor:MYOPIC_static_arrival_rates}
	Under Assumption \ref{ass:zipf}, the Block Change Model defined in
	Section \ref{subsec:change_model}, and for $\beta > 1$, over any time
	interval $T(n)$ such that $T(n) = \Omega(1)$ and $T(n) =
	\oo(n^{\beta-1})$, the expected number of requests deferred by
	any learning-based static storage policy is  $\Omega\big(\min\{(nT(n))^{\frac{1}{2-1/\beta}},n\}\big)$ and the
	number of requests deferred by the MYOPIC policy is $\OO\big((nT(n))^{\frac{1}{\beta}}\big)$
	w.h.p.
\end{corollary}

For $\beta > 1$, $\frac{1}{2-1/\beta} > \frac{1}{\beta}$ and for $T(n) =
\oo(n^{\beta-1})$, $(nT(n))^{\frac{1}{\beta}} = \oo(n)$ . Therefore,
from Corollary \ref{cor:MYOPIC_static_arrival_rates}, we conclude that
MYOPIC outperforms all learning-based static storage policies. Note
that:
\begin{itemize}
	\item[i.] Corollary \ref{cor:MYOPIC_static_arrival_rates} holds even
	when the interval of interest $T(n)$ grows to infinity (scaling
	polynomially in $n$), or correspondingly, even when the content-type
	popularity changes very slowly with time.
	
	\item[ii.] Even if the partitioning of the $(T(n))$ into a Learning
	phase and a Static Storage phase is done in an optimal manner with
	the help of some side information $(\beta, T(n))$, the MYOPIC
	algorithm outperforms any learning-based static storage policy.
	
	\item[iii.] Since we consider the high-dimensional setting, the
	learning problem at hand is a large-alphabet learning problem. It is
	well known that standard estimation techniques like using the
	empirical values as estimates of the true statistics is suboptimal
	in this setting. Many learning algorithm like the classical
	Good-Turing estimator \cite{MS00} and other linear estimators
	\cite{VV11} have been proposed, and shown to have good performance
	for the problem of large-alphabet learning. From Corollary
	\ref{cor:MYOPIC_static_arrival_rates}, we conclude that, even if the
	learning-based storage policy uses the best possible large-alphabet
	estimator, it cannot match the performance of the MYOPIC policy.
\end{itemize}


Therefore, in the high-dimensional setting we consider, separating the
task of estimation of the  demand statistics, and the subsequent use
of the same to design a static storage policy, is strictly
suboptimal. This is the key message of this paper.\\

Theorem \ref{thm:MYOPIC_static_arrival_rates} characterizes the
performance of MYOPIC under the Block Change Model, where the
statistics of the arrival process do not change in interval of
interest. To gain further insight into robustness of MYOPIC against
changes in the arrival process, we now analyze the performance of
MYOPIC when the arrival process can change in the interval of interest
according to the Continuous Change Model defined in
Section~\ref{subsec:change_model}.

Recall that under the Continuous Change Model, on average, we expect
$\Theta(n)$ shuffles in the popularity of various content-types in an
interval of constant duration. For the Block Change Model, if $T(n) =
\Theta(1)$, the entire popularity distribution can change at the end
of the block, which is equivalent to $n$ shuffles. Therefore, for both
the change models, the expected number of changes to the popularity
distribution in an interval of constant duration is of the same
order. However, these changes occur constantly but slowly in the
Continuous Change Model as opposed to a one-shot change in the Block
Change Model.

\begin{theorem}
	\label{thm:MYOPIC_changing_arrival_rates}
	Under Assumption \ref{ass:zipf}, and the Continuous Change Model
	defined in Section \ref{subsec:change_model}, the number of requests
	deferred by the MYOPIC storage policy in any interval of constant
	duration is $\OO(n^{1/\beta})$ w.h.p.
\end{theorem}

In view of Theorem \ref{thm:MYOPIC_static_arrival_rates}, if the
arrival rates do not vary in an interval of constant duration, under
the MYOPIC storage policy, the number of requests deferred in that
interval is $\OO(n^{1/\beta})$ w.h.p. Theorem
\ref{thm:MYOPIC_changing_arrival_rates} implies that the number
of requests deferred in a constant duration interval is of the same
order even if the arrival rates change according to the Continuous
Change Model. This shows that the performance of the MYOPIC policy is
robust to changes in the popularity statistics.

We now compare the upper bound obtained in Theorem \ref{thm:MYOPIC_changing_arrival_rates} for the Continuous Change Model with the lower bound on the performance of any learning-based static storage policy obtained in Theorem \ref{thm:converse_2}.

\begin{corollary}
	\label{cor:MYOPIC_changing_arrival_rates}
	Under Assumption \ref{ass:zipf}, the Continuous Change Model defined in
	Section \ref{subsec:change_model}, and for $\beta > 1$, over any time
	interval of constant duration, the expected number of requests deferred by
	any learning-based static storage policy is $\Omega\big(n^{\frac{1}{2-1/\beta}}\big)$ and the number of requests deferred by the MYOPIC policy is $\OO(n^{\frac{1}{\beta}})$
	w.h.p.
\end{corollary}

Thus, even for the Continuous Change Model, MYOPIC outperforms all Learning-based static policies.  Compared to the Block Change Model, Learning-based static policies are ``unsuitable" for the Continuous Change Model due to the following reasons:
\begin{itemize}
	\item[-] Content popularity can change while the system is in the learning phase. This makes the task of estimating content popularity more difficult.
	\item[-] Once storage is optimized for the estimated content popularity (at the end of Phase 1), it is not changed in Phase 2. However, content popularities will change (by a small amount) almost instantaneously after the learning period, thus making the storage suboptimal even if content popularity was estimated accurately in Phase 1.
\end{itemize}

\subsection{Genie-Aided Optimal Storage Policy}
\label{subsec:optimal}

In this section, our objective is to study the setting where the
demand statistics are available ``for free''. For the Block Change Model with known popularity statistics, we show that a simple
adaptive policy is optimal in the class of all policies which know
popularity statistics of various content-types. We denote the class of
such policies as $\mathds{A}$ and refer to the optimal policy as the
GENIE policy.


\noindent Let the content-types be indexed from $i = 1$ to $m$ and let
$C_i$ be the $i^{th}$ content-type. Without loss of generality, we
assume that the content-types are indexed in the order of popularity,
i.e, $\lambda_i \geq \lambda_{i+1}$ for all $i \geq 1$. Let $\idle
(t)$ denote the number of idle servers at time $t$.

The key idea of the GENIE storage policy is to ensure that at any time
$t$, if the number of idle servers is $\idle (t)$, the $\idle (t)$
most popular content-types are stored on exactly one idle server
each. The GENIE storage policy can be implemented as follows. Recall
$C_i$ is the $i^{th}$ most popular content-type. At time $t$,
\begin{itemize}
	\item[-] If there is a request for content-type $C_i$ with $i <
	k(t^-),$ then allocate the request to the corresponding idle
	server. Further, replace the content-type on server storing
	$C_{k(t^-)}$ with content-type $C_i.$
	
	\item[-] If there is a request for content-type $C_i$ with $i >
	k(t^-),$ defer this request. There is no storage update.
	
	\item[-] If there is a request for content-type $C_i$ with $i =
	k(t^-),$ then allocate the request to the corresponding idle
	server. There is no storage update.
	
	\item[-] If a server becomes idle (due to a departure), replace its
	content-type with $C_{\idle (t^-)+1}$.
\end{itemize}

\noindent For a formal definition, please refer to Figure
\ref{policy:GENIE}. \\

\begin{figure}[h]
	\hrule
	\vspace{0.1in}
	\begin{algorithmic}[1]
		\STATE Initialize: Number of idle-servers $:= \idle = n$.
		\WHILE {true}
		\IF {new request (for $C_i$) routed to a server,}
		\IF {$i \neq \idle$,}
		\STATE replace content-type of idle server storing $C_\idle$ with $C_i$
		\ENDIF
		\STATE $\idle \gets \idle-1$
		\ENDIF
		\IF{departure,}
		\STATE replace content-type of new idle server with $C_{\idle+1}$
		\STATE $\idle \gets \idle +1$
		\ENDIF
		\ENDWHILE
	\end{algorithmic}
	\vspace{0.1in}
	\hrule
	\caption{GENIE -- \sl An adaptive storage policy which has content popularity statistics available for ``free". At time $t$, if the number of idle servers is $\idle(t)$, the $\idle(t)$ most popular content-types are stored on exactly one idle server each.}
	\label{policy:GENIE}
\end{figure}
\begin{remark} The implementation of GENIE requires replacing the
	content-type of at most one server on each arrival and departure.
\end{remark}

To characterize the performance of GENIE, we assume that the system
starts from the empty state (all servers are idle) at time $t=0$. The
performance metric for any policy $\mathcal{A}$ is
$D^{(\mathcal{A})}(t)$, defined as the number of requests deferred by
time $t$ under the adaptive storage policy $\mathcal{A}$.  We say that
an adaptive storage policy $\mathcal{O}$ is optimal if
\begin{eqnarray}
\label{eq:stochastic_dominance}
D^{(\mathcal{O})}(t) \leq_{st} D^{(\mathcal{A})}(t),
\end{eqnarray}
for any storage policy $\mathcal{A} \in \mathds{A}$ and any time $t \geq 0$.
Where Equation \ref{eq:stochastic_dominance} implies that,
\begin{eqnarray*}
	\mathbb{P}(D^{(\mathcal{O})}(t)>x) \leq \mathbb{P}(D^{(\mathcal{A})}(t)>x),
\end{eqnarray*}
for all $x \geq 0$ and $t \geq 0$.

\begin{theorem}
	\label{thm:stochastic_dominance}
	If the arrival process to the content-type delivery system is Poisson
	and the service times are exponential random variables with mean 1,
	for the Block Change Model defined in Section
	\ref{subsec:change_model}, let $D^{(\mathcal{A})}(t)$ be the number of
	requests deferred by time $t$ under the adaptive storage policy
	$\mathcal{A} \in \mathds{A}$. Then, we have that,
	\begin{eqnarray*}
		D^{(GENIE)}(t) \leq_{st} D^{(\mathcal{A})}(t),
	\end{eqnarray*}
	for any storage policy $\mathcal{A} \in \mathds{A}$ and any time $t \geq 0$.
\end{theorem}

Note that this theorem holds even if the $\lambda_i$s are not
distributed according to the Zipf distribution. We thus conclude that
GENIE is the optimal storage policy in the class of all storage
policies which at time $t$, have no additional knowledge of the future
arrivals except the values of $\lambda_i$ for all content-types and
the arrivals and departures in $[0, t)$. Next, we compute a lower
bound on the performance of GENIE.

\begin{theorem}
	\label{thm:GENIE_static_arrival_rates}
	Under Assumption \ref{ass:zipf}, for $\beta > 1$, the Block Change Model defined in Section \ref{subsec:change_model} and if the interval
	of interest is of constant length, the expected number of requests
	deferred by GENIE is $\Omega(n^{2-\beta})$.
\end{theorem}

From Theorems \ref{thm:MYOPIC_static_arrival_rates} and
\ref{thm:GENIE_static_arrival_rates} we see that there is a gap in the
performance of the MYOPIC policy and the GENIE policy (which has
additional knowledge of the content-type popularity statistics). Since
for the GENIE policy, learning the statistics of the arrival process
comes for ``free'', this gap provides an upper bound on the cost of
serving content-type with \emph{unknown} demands. We compare the
performance of the all the policies considered so far in the next
section via simulations.

%
%
%

As discussed before, the key property of the GENIE storage policy is that at time $t$, if there are $\idle(t)$ idle servers, the policy ensures that exactly one copy of the $\idle(t)$ most popular contents is stored on the idle servers. In Figure \ref{policy:GENIE}, we describe how to preserve this property at all times, in the setting where content popularity remains constant in the interval of interest.  If content popularity is time-varying, as in the case of the Continous Change Model, to maintain this property, the policy needs to have instantaneous knowledge of any change in content popularity. Moreover, contents stored on idle servers might need to be changed at the instant of change in content popularity to ensure that the idle servers store the currently most popular contents at all times.

Since the MYOPIC and GENIE policies are adaptive policies, contents stored on the front-end servers are changed dynamically. Such content changes can be classified into two types: internal fetches and external fetches. An internal fetch occurs when a content is available on at least one front-end server and the storage policy needs to place a copy of this content on an idle front-end server. In such cases, we assume that the new copy is fetched internally from one of the local (front-end) servers storing this content. An external fetch occurs when the content is currently not stored on any of the front-end servers (busy/idle) and hence the copy needs to be fetched externally from the back-end server. The external fetches incur a much higher cost compared to the internal fetches as data transfer from outside is subject to high delay and/or bandwidth consumption. The next theorem provides bounds on the number of external fetches performed to implement the MYOPIC and GENIE policies under the Block Change Model. Since the comparison depends on the initial storage of servers at the beginning of the block, we consider the worst initial case for the MYOPIC policy which is an empty system.
\begin{theorem}
	\label{thm:adaptation_cost}
	Let $V^{P^*}(T)$  be the number of external fetches made while implementing the storage policy $P^*$ in the time-interval $(0,T)$. Under Assumption 1, for $\beta > 1$, the Block Change Model and assuming we start from an empty system, for $T = \OO(1)$,
	\begin{itemize}
		\item [(i)] $V^{(\text{MYOPIC})}(T) = \OO(nT)^{1/\beta}$ w.h.p.
		\item [(ii)] $V^{(\text{GENIE})}(T) = \Omega \{\min \{n, nT \}$\} w.h.p.
	\end{itemize}
\end{theorem}

Thus the MYOPIC policy incurs fewer external fetches compared to the GENIE policy. This is not surprising as the GENIE storage policy is designed with the objective of minimizing the number of deferred requests, and hence it is more aggressive in changing the contents stored on servers in order to minimize the probability that the next request is deferred.

\section{Simulation Results}
\label{sec:simulation_results}

We compare the performance of the MYOPIC policy with the performance of the GENIE policy and the following two learning-based static storage policies:
\begin{itemize}
	\item [-] The \emph{``Empirical + Static Storage''} policy uses the empirical popularity statistics of content types in the learning phase as estimates of the the true popularity statistics. At the end of the learning phase, the number of servers on which a content is stored is proportional to its estimated popularity.
	\item [-] The \emph{``Good Turing + Static Storage''} policy uses the
	Good-Turing estimator \cite{MS00} to compute an estimate of the
	missing mass at the end of the learning phase. The missing mass is
	defined as total probability mass of the content types that were not
	requested in the learning phase. Recall that we assume that learning-based static storage policies treat all the missing content-types equally, i.e., all missing content-types are estimated to be equally popular.
	
	Let $M_0$ be the total probability mass of the content types that were not requested in the learning
	phase and $S_1$ be the set of content types which were requested
	exactly once in the learning phase. The Good-Turing estimator of
	the missing mass $(\widehat{M}_0)$ is given by
	$$\widehat{M}_0 = \frac{|S_1|}{\text{number of samples}}.$$
	See \cite{MS00} for details.
	
	Let $N_i$ be the number of times content $i$ was requested in the
	learning phase and $\mathcal{C}_{\text{missing}}$ be the set of
	content-types not requested in the learning phase. The ``Good Turing + Static Storage'' policy computes an estimate of the content-popularity as follows:
	\begin{itemize}
		\item[i:] If $N_i = 0$, $p_i =  \dfrac{\widehat{M}_0}{|\mathcal{C}_{\text{missing}}|}$.
		\item[ii:] If $N_i > 0$, $p_i = (1-\widehat{M}_0) \dfrac{N_i}{\text{number of samples}}$.
	\end{itemize}
	At the end of the learning phase, the number of servers on which a content is stored is proportional to its estimated popularity.
\end{itemize}

We simulate the content distribution system for arrival and service process which satisfy Assumption \ref{ass:zipf} to compare the performance of the four policies mentioned above and also understand how their performance depends on various parameters like system size $(n)$, load $(\bar{\lambda})$ and Zipf parameter $(\beta)$. In Tables \ref{table:diff_n}, \ref{table:diff_beta} and \ref{table:diff_load}, we report the mean and variance of the fraction of jobs served by the policies over a duration of 5~s ($T(n) = 5$).

For each set of system parameters, we repeat the simulations between 1000 to 10000 times for each policy in order to ensure that the standard deviation of the quantity of interest (fraction of jobs served) is small and comparable. For the two adaptive policies (GENIE and MYOPIC), the results are averaged over 1000 iterations and for the learning-based policies (``Empirical + Static Storage'' and ``Good-Turing + Static Storage''), the results are averaged over 10000 iterations. In addition, the results for the learning-based policies are reported for empirically optimized values for the fraction of time spent by the policy in learning the distribution.


In Table \ref{table:diff_n}, we compare the performance of the policies for different values of system size ($n$). For the results reported in Table \ref{table:diff_n}, the ``Empirical + Static Storage'' policy learns for 0.1~s and the ``Good Turing + Static Storage'' policy learns for 0.7~s. The performance of all four policies improves as the system size increases and the adaptive policies significantly outperform the two learning-based static storage policies. Figure \ref{fig:diff_n} is a plot of the mean values reported in Table \ref{table:diff_n}.

\begin{table}[h]
	
	\centering
	\begin{tabular}{l r c c }
		\hline\hline
		Policy & $n$ & Mean & $\sigma$ \\ [0.5ex] 
		\hline 
		GENIE & 200 & 0.9577 & 0.0081  \\
		& 400 & 0.9698 & 0.0045 \\
		& 600 & 0.9752 & 0.0034 \\
		& 800 & 0.9788 & 0.0030 \\
		& 1000 & 0.9814 & 0.0025 \\
		\hline
		MYOPIC & 200 & 0.8995  & 0.0258 \\
		& 400 & 0.9260  & 0.0167  \\
		& 600 & 0.9380  & 0.0132 \\
		& 800 & 0.9481  & 0.0101 \\
		& 1000 & 0.9532 & 0.0080 \\
		\hline
		Empirical + Static Storage & 200 & 0.6292 & 0.0662  \\
		& 400 & 0.6918 & 0.0443 \\
		& 600 & 0.7246 & 0.0353 \\
		& 800 & 0.7464 & 0.0304 \\
		& 1000 & 0.7622 & 0.0268 \\
		\hline
		Good Turing + Static Storage & 200 & 0.6875 & 0.0274  \\
		& 400 & 0.7249 & 0.0180  \\
		& 600 & 0.7443 & 0.0140 \\
		& 800 & 0.7566 & 0.0118 \\
		& 1000 & 0.7651 & 0.0104 \\
		\hline
	\end{tabular}
	\caption{\sl The performance of the four policies as a function of the system size $(n)$ for fixed values of load $\bar{\lambda} = 0.8$ and $\beta = 1.5$. The values reported are the mean and standard deviation ($\sigma$) of the fraction of jobs served. Both adaptive policies (GENIE and MYOPIC) significantly outperform the two learning-based static storage policies.}
	
	\label{table:diff_n} 
\end{table}

\begin{figure}[h]
	\begin{center}
		\includegraphics[scale=0.35]{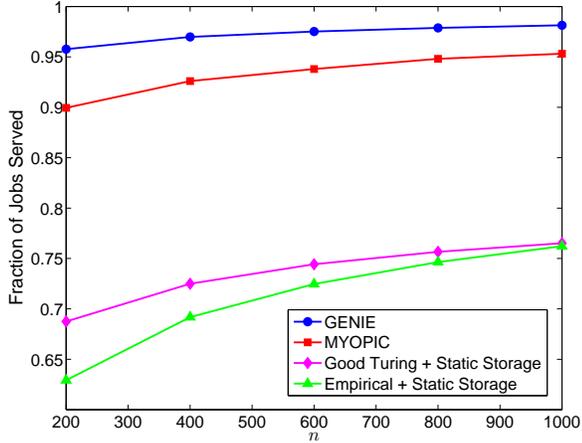}
		\caption{\sl Plot of the mean values reported in Table \ref{table:diff_n} -- performance of the storage policies as a function of system size $(n)$ for $\bar{\lambda} = 0.8$ and $\beta = 1.5$. \label{fig:diff_n}}
	\end{center}
\end{figure}

In Table \ref{table:diff_beta}, we compare the performance of the policies for different values of Zipf parameter $\beta$. For the results reported in Table \ref{table:diff_beta}, the duration of the learning phase for both learning based policies is fixed such that the expected number of arrivals in that duration is 100. The performance of all four policies improves as the value of the Zipf parameter $\beta$ increases, however, the MYOPIC policy outperforms both learning-based static storage policies for all values of $\beta$ considered.

\begin{table}[h]
	\centering
	\begin{tabular}{l c c c }
		\hline\hline
		Policy & $\beta$ & Mean & $\sigma$ \\ [0.5ex] 
		\hline 
		GENIE & 2 & 0.9939 & 0.0026 \\
		& 3 & 0.9996 & 0.0015 \\
		& 4 & 0.9998 & 0.0011 \\
		& 5 & 0.9998 & 0.0012 \\
		& 6 & 0.9998 & 0.0011 \\
		\hline
		MYOPIC & 2 & 0.9778 & 0.0078 \\
		& 3 & 0.9960 & 0.0033 \\
		& 4 & 0.9982 & 0.0026 \\
		& 5 & 0.9990 & 0.0018 \\
		& 6 & 0.9993 & 0.0013 \\
		\hline
		Empirical + Static Storage & 2 & 0.8594 & 0.0194  \\
		& 3 & 0.9228 & 0.0155 \\
		& 4 & 0.9397 & 0.0119 \\
		& 5 & 0.9453 & 0.0095 \\
		& 6 & 0.9495 & 0.0073 \\
		\hline
		Good Turing + Static Storage & 2 & 0.8436 & 0.0235  \\
		& 3 & 0.9198 & 0.0154 \\
		& 4 & 0.9378 & 0.0124 \\
		& 5 & 0.9456 & 0.0094 \\
		& 6 & 0.9491 & 0.0072 \\
		\hline
	\end{tabular}
	\caption{\sl The performance of the four policies as a function of the Zipf parameter $(\beta)$ for fixed values of system size $n = 500$ and load $\bar{\lambda} = 0.9$. The values reported are the mean and standard deviation ($\sigma$) of the fraction of jobs served. The MYOPIC policy outperforms the two learning-based static storage policies for all values of $\beta$ considered.}
	
	\label{table:diff_beta} 
\end{table}

In Table \ref{table:diff_load}, we compare the performance of the policies for different values of load $\bar{\lambda}$. For the results reported in Table \ref{table:diff_load}, the duration of the learning phase for both learning based policies is fixed such that the expected number of arrivals in that duration is 100. The performance of all four policies deteriorates as the load increases, however, for all loads considered, the MYOPIC policies outperforms the two learning-based static storage policies.

\begin{table}[h]
	\centering
	\begin{tabular}{l c c c }
		\hline\hline
		Policy & $\bar{\lambda}$ & Mean & $\sigma$ \\ [0.5ex] 
		\hline 
		GENIE & 0.500 & 0.9892 & 0.0025 \\
		& 0.725 & 0.9788 & 0.0013 \\
		& 0.950 & 0.9531 & 0.0017 \\
		\hline
		MYOPIC & 0.500 & 0.9605 & 0.0113 \\
		& 0.725 & 0.9484 & 0.0105 \\
		& 0.950 & 0.8973 & 0.0221 \\
		\hline
		Empirical + Static Storage & 0.500 & 0.7756 & 0.0222  \\
		& 0.725 & 0.7705 & 0.0238 \\
		& 0.950 & 0.7352 & 0.0235 \\
		\hline
		Good Turing + Static Storage & 0.500 & 0.7849 & 0.0230  \\
		& 0.725 & 0.7589 & 0.0249 \\
		& 0.950 & 0.6869 & 0.0348 \\
		\hline
	\end{tabular}
	\caption{\sl The performance of the four policies as a function of the load $(\bar{\lambda})$ for fixed values of system size $n = 500$ and $\beta = 1.2$. The values reported are the mean and standard deviation ($\sigma$) of the fraction of jobs served. The MYOPIC policy significantly outperforms the two learning-based static storage policies for all loads considered.}
	
	\label{table:diff_load} 
\end{table}

In Figure \ref{fig:adaptation_cost}, we plot the mean value (with error bars of 3$\times$std. dev.) of the number of external fetches made by the MYOPIC and GENIE storage policies for different values of $n$ and $\beta$ for a load of 0.9 averaged over 10000 iterations. As expected, the GENIE storage policy makes more external fetches than the MYOPIC policy.

\begin{figure}[h]
	\begin{center}
		\includegraphics[scale=0.35]{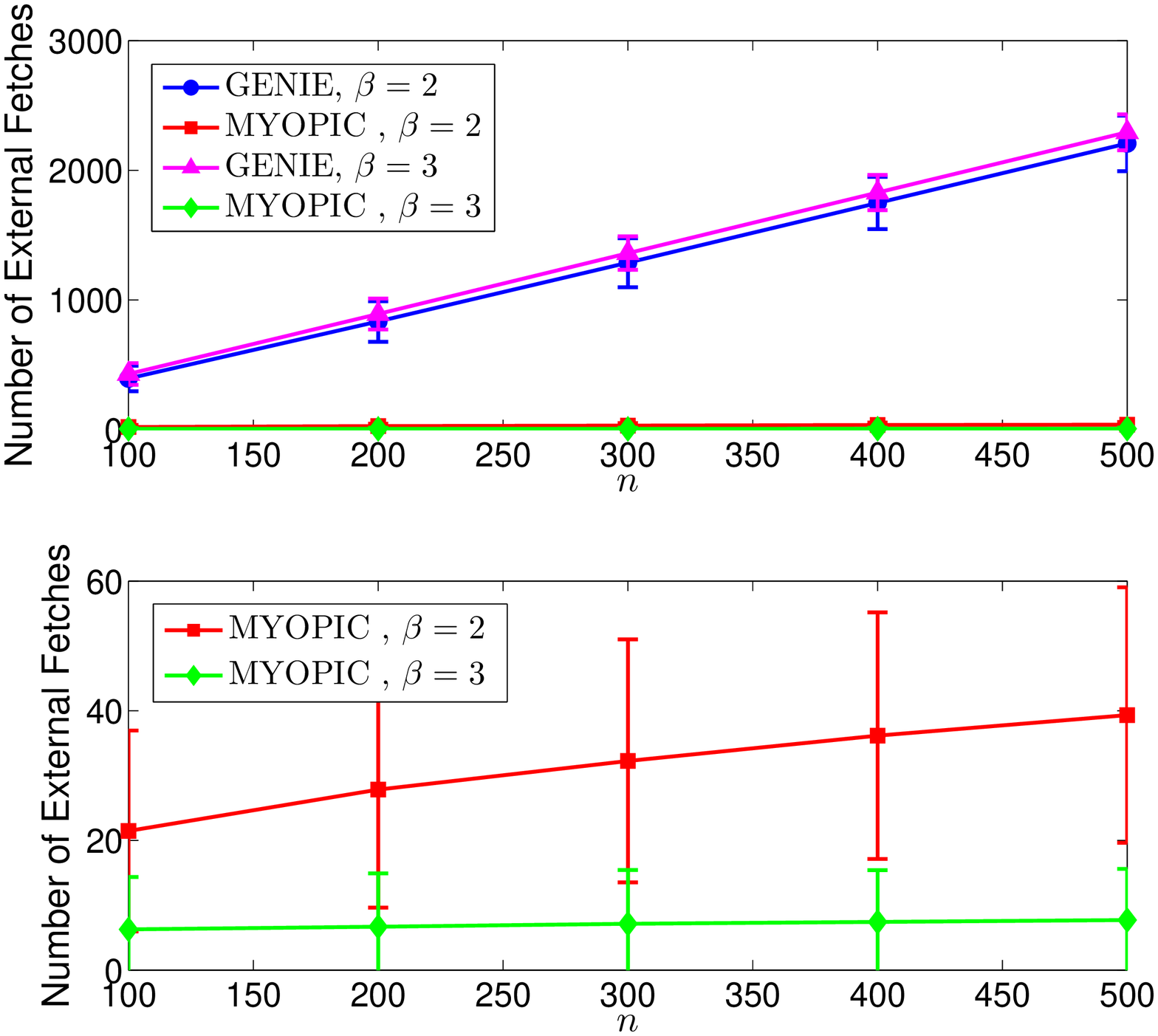}
		\caption{\sl The mean number of external fetches (content fetched from the back-end server to place on a front-end server) by the two adaptive policies as a function of system size $(n)$ for $\bar{\lambda} = 0.9$ and $\beta = 2$ and $3$. The first plot shows the performance of both GENIE and MYOPIC. The second plot focuses only on the performance of the MYOPIC storage policy for clarity. \label{fig:adaptation_cost}}
	\end{center}
\end{figure}

\section{Proofs of Main Results}
\label{proofs}
In this section, we provide the proofs of our results. 
\subsection{Proof of Theorem 1}
\label{sec:proof2}

\noindent We first present an outline of the proof of Theorem \ref{thm:converse}. We consider two cases. We first focus on the case when the learning-based storage policies use fewer than $n$ arrivals to learn the distribution.
\begin{enumerate}
	\item If the learning phase lasts for the first $n^{\gamma}$ arrivals for some $0<\gamma \leq 1$, we show that under Assumption \ref{ass:zipf}, w.h.p., in the learning phase, there are no arrivals for at least $n - \OO(n^{\frac{\gamma}{\beta}})$ content types. (Lemma \ref{lemma:E_1}).
	\item Next, we show that w.h.p., among the first $n^{\gamma}$ arrivals, i.e., during the learning phase, $\Omega(n^{\gamma})$ requests are deferred (Lemma \ref{lemma:E_2}).
	\item Using Lemma \ref{lemma:E_1}, we compute a lower bound on the number of requests deferred in Phase 2 (after the learning phase) by any learning-based static storage policy (Lemma \ref{lemma:phase2_drops}).
	\item Using Steps 2 and 3, we lower bound the number of requests deferred in the interval of interest.
\end{enumerate}

In the case when the learning phase lasts for more than $n$ arrivals, we show that the number of requests deferred in the learning phase alone is $\Omega(n)$, thus proving the theorem for this case.

\begin{lemma}
	\label{lemma:E_1}
	
	Let $E_1$ be the event that in the first $n^{\gamma}$ arrivals, for $0 < \gamma < 1$ no more than $\OO(n^{\frac{\gamma}{\beta}})$ different types of contents are requested. Then,
	\begin{eqnarray}
	\label{eq:E_1}
	\mathbb{P}(E_1^c) = \oo\bigg(\frac{1}{n}\bigg).
	\end{eqnarray}
	for $n$ large enough.
\end{lemma}
\begin{proof}
	Recall $\lambda_i = \bar{\lambda}n p_i$ where $p_i = \frac{i^{-\beta}}{Z(\beta)}$ for $Z(\beta) = \sum_{i=1}^m i^{-\beta}$.
	\begin{eqnarray*}
		Z(\beta) = \sum_{i=1}^{\alpha n} i^{-\beta} &\geq& \int_{1}^{\alpha n+1} i^{-\beta} di
		\geq \frac{0.9}{\beta-1}
	\end{eqnarray*}
	for $n$ large enough.
	Therefore, for all $i$,
	\begin{eqnarray*}
		p_i \leq \frac{\beta-1}{0.9} i^{-\beta}.
	\end{eqnarray*}
	The total mass of all content types $i = k, .. m = \alpha n$ is
	\begin{eqnarray*}
		\sum_{i=k}^{\alpha n} p_i \leq \sum_{i=k}^{\alpha n} \frac{\beta-1}{0.9} i^{-\beta}
		\leq \int_{k-1}^{\alpha n} \frac{\beta-1}{0.9} i^{-\beta} di
		\leq \frac{1}{0.9} \dfrac{1}{(k-1)^{\beta-1}}.
	\end{eqnarray*}
	
	Now, for $k = (n)^{\gamma/\beta}$ + 1, we have that,
	\begin{eqnarray*}
		\sum_{i=k}^{\alpha n} p_i  \leq \frac{1}{0.9} \frac{n^{\gamma/\beta}}{n^{\gamma}}.
	\end{eqnarray*}
	Therefore, the expected number of requests for content types $k, k+1, .. \alpha n$ is less than $\frac{1}{0.9} (n^{\gamma/\beta})$. Using the Chernoff bound, the probability that there are more than $\frac{2}{0.9} (n^{\gamma/\beta})$ requests for content types $k, k+1, .. \alpha n$ in the interval of interest is less than $ \frac{1}{n^2}$ for $n$ large enough.
	
	Therefore, with probability greater than $ 1-1/n^2$, the number different types of contents requests for in the interval of interest is less than $ n^{\gamma/\beta} + \frac{2}{0.9} (n^{\gamma/\beta})$. Hence the result follows.
\end{proof}

We use the following concentration result for Exponential random variables.
\begin{lemma}
	\label{lemma:sum_of_exponentials}
	Let $X_k$ for $0\leq k \leq v$, be i.i.d. exponential random variables with mean 1, then,
	\begin{eqnarray}
	\label{eq:exp}
	\mathbb{P}\bigg(\sum_{k=1}^{v} X_i \leq a\bigg) \leq  \exp(v-a)\bigg(\frac{a}{v}\bigg)^v.
	\end{eqnarray}
\end{lemma}
\begin{proof}
	This follows from elementary calculations, and is provided here for completeness. For any $a$ and $v$, by the Chernoff bound, we have that,
	\begin{eqnarray*}
		\mathbb{P}\bigg(\sum_{k=1}^{v} X_i \leq a\bigg) \leq \min_{t > 0} e^{ta} (E[e^{-tX_i}])^v.
	\end{eqnarray*}
	Since $X_k$ is an exponential random variable with mean 1, we have that,
	\begin{eqnarray*}
		\mathbb{P}\bigg(\sum_{k=1}^{v} X_i \leq a\bigg) \leq \min_{t > 0} e^{ta} \bigg(\frac{1}{1+t} \bigg)^v = \exp(v-a)\bigg(\frac{a}{v}\bigg)^v.
	\end{eqnarray*}
\end{proof}

\begin{lemma}
	\label{lemma:E_2}
	Suppose the system starts with each content piece stored on exactly one server. Let $E_2$ be the event that in the first $n^{\gamma}$ arrivals for $\gamma$ such that $0 < \gamma < 1$, at most $\OO(n^{\gamma/\beta})(\log n+1) $ are served (not deferred). Then, for $\beta > 1$,
	\begin{eqnarray}
	\label{eq:E_2}
	\mathbb{P}(E_2) \geq 1- \frac{1}{\log n}.
	\end{eqnarray}
\end{lemma}
\begin{proof}
	This proof is conditioned on the event $E_1$ defined in Lemma \ref{lemma:E_1}. Conditioned on $E_1$, in the first $n^{\gamma}$ arrivals, at most $\OO(n^{\gamma/\beta})$ different content types are requested. Therefore, at most $\OO(n^{\gamma/\beta})$ servers can serve requests during the first $n^{\gamma}$ arrivals.
	
	Let $E_3$ be the event that the time taken for the first $n^{\gamma}$ arrivals is less than $\frac{2n^{\gamma}}{\bar{\lambda} n}$. Since the expected time for the first $n^{\gamma}$ arrivals is $\frac{n^{\gamma}}{\bar{\lambda} n}$, by the Chernoff bound, $\mathbb{P}(E_3) \geq 1-\oo(1/n)$. The rest of this proof is conditioned on the event $E_3$.
	
	If the system serves (does not defer) more than $\OO(n^{\gamma/\beta}(\log n+1))$ requests in this interval, at least one server needs to serve more than $\log n$ requests. By substituting $a = cn^{-1+\gamma}$ and $v = \log n$ in Lemma \ref{lemma:sum_of_exponentials}, we have that,
	\begin{eqnarray*}
		\mathbb{P}\bigg(\sum_{k=1}^{\log n} X_k \leq cn^{-1+\gamma}\bigg) &\leq&  \exp(\log n-cn^{-1+\gamma})\\
		&&\times \bigg(\frac{cn^{-1+\gamma}}{\log n}\bigg)^{\log n} = \oo \bigg(\frac{1}{n}\bigg).
	\end{eqnarray*}
	Therefore, the probability that a server serves more than $\log n$ requests in an interval of $\frac{2n^{\gamma}}{\bar{\lambda} n}$ time is $ \oo \big(\frac{1}{n}\big)$. Therefore, using the union bound, the probability that none of these $\OO(n^{\gamma/\beta})$ servers serve more than $\log n$ requests each in $\frac{2n^{\gamma}}{\bar{\lambda} n}$ time is greater than $ 1 - \OO(n^{\gamma/\beta})\oo(\frac{1}{n})$.
	Therefore, we have that,
	\begin{eqnarray*}
		\mathbb{P}(E_2^c) &\leq& \OO(n^{\gamma/\beta}) \oo\bigg(\frac{1}{n}\bigg) + P(E_1^c) + P(E_3^c)\\
		&\leq& \frac{1}{\log n}
	\end{eqnarray*}
	for $n$ large enough.
\end{proof}

\begin{lemma}
	\label{lemma:phase2_drops}
	Let the interval of interest be $T(n)$ such that $T(n) = \Omega(1)$. If the learning phase of the storage policy lasts for the first $n^{\gamma}$ arrivals, $0<\gamma < 1$, the expected number of requests deferred in Phase 2 is $\Omega\big(T(n)n^{1-\gamma}n^{\frac{\gamma}{\beta}}\big)$.
\end{lemma}
\begin{proof}
	Let $N_2$ be the number of arrivals in Phase 2, then we have that,
	$E[N_2] = T(n)\bar{\lambda} n - n^{\gamma}$.
	
	Let $E_4$ be the event that $N_2 > E[N_2]/2$. Using the Chernoff bound, it can be shown that $P(E_4^c) = \oo(1/n)$.
	
	The rest of this proof is conditioned on $E_1$ defined in Lemma \ref{lemma:E_1} and $E_4$ defined above. We consider the following two cases depending on the number of servers allocated to content types not seen in Phase 1.\\
	
	\noindent Case I: The number of servers allocated to content types not seen in Phase 1 is less than $\epsilon n$ for some $\epsilon \leq 1-\frac{\bar{\lambda}}{1000}$. For $\beta > 1$,
	
	\begin{eqnarray*}
		Z(\beta) = \sum_{i=1}^{ \alpha n} i^{-\beta} \leq \sum_{i=1}^{\infty} i^{-\beta} = c_z < \infty.
	\end{eqnarray*}
	Therefore, for all $i$, $p_i \geq \frac{1}{c_z} i^{-\beta}.$
	The total mass of all content types $k,k+1, .. \alpha n$ is
	\begin{eqnarray*}
		\sum_{i=k}^{\alpha n} p_i &\geq& \sum_{i=k}^{\alpha n} \frac{1}{c_z} i^{-\beta} \geq \int_{k}^{\alpha n+1} \frac{1}{c_z} i^{-\beta} di\\
		&=& \frac{0.9}{c_z (\beta -1)} \dfrac{1}{k^{\beta-1}},
	\end{eqnarray*}
	for $n$ large enough.
	
	Therefore, the expected number of arrivals of types not requested in Phase 1 in Phase 2 is at least $(\frac{T(n) \bar{\lambda} n- n^{\gamma}}{2})\frac{0.9}{c_z (\beta -1)} \frac{n^{\frac{\gamma}{\beta}}}{n^{\gamma}}$.
	
	Let $E_5$ be the event that in Phase 2, there are at least $(\frac{T(n) \bar{\lambda} n- n^{\gamma}}{4})\frac{0.9}{c_z (\beta -1)} \frac{n^{\frac{\gamma}{\beta}}}{n^{\gamma}}$ arrivals of types not requested in Phase 1. Using the Chernoff bound,
	$
	\mathbb{P}(E_5^c) = \oo(1/n).
	$
	
	Conditioned on $E_1$, all but $\OO(n^{{\gamma}/{\beta}})$ content types, are not requested in Phase 1. Recall that all learning-based policies treat all these content types equally and that the total number of servers allocated to store the content types not seen in Phase 1 is less than $ \epsilon n$. Let $\eta$ be the probability that a content is not stored by the storage policy under consideration. Then,
	\begin{eqnarray*}
		\eta \geq 1 - \frac{\epsilon n}{n - \OO(n^{{\gamma}/{\beta}})} \geq 1- \frac{\epsilon}{2},
	\end{eqnarray*}
	for $n$ large enough. \\
	
	Let $E_6 = E_1 \cap E_3 \cap E_4 \cap E_5$ and $D_2$ be the number of requests deferred in Phase 2.
	\begin{eqnarray*}
		E[D_2 | E_6] &\geq& \eta \bigg(\bigg(\frac{T(n) \bar{\lambda} n- n^{\gamma}}{2}\bigg)\frac{0.9}{2 c_z (\beta -1)} \frac{n^{\gamma/\beta}}{n^{\gamma}}\bigg)\\
		&\geq& \bigg(1-\frac{\epsilon}{2}\bigg) \bigg(\frac{T(n) \bar{\lambda} n- n^{\gamma}}{2}\bigg)\frac{0.9}{2 c_z (\beta -1)} \frac{n^{\gamma/\beta}}{n^{\gamma}} \\
		& = & \Omega \big(T(n)n^{1-\gamma}n^{\gamma/\beta}\big).
	\end{eqnarray*}
	Therefore,
	\begin{eqnarray*}
		E[D_2] &\geq& E[D_2 | E_6] \mathbb{P}(E_6)\\
		&\geq& E[D_2 | E_6] \bigg(1- \dfrac{1}{\log n} - \dfrac{3}{n}\bigg) \\
		&=& \Omega \big(T(n)n^{1-\gamma}n^{\gamma/\beta}\big). \\
	\end{eqnarray*}
	
	\noindent Case II: The number of servers allocated to content types not seen in Phase 1 is more than $\epsilon n$ for some $\epsilon > 1-\frac{\bar{\lambda}}{1000}$.
	
	Let $f(n)$ be the number of servers allocated to store all content types that are requested in Phase 1. By our assumption, $f(n) \leq \frac{\bar{\lambda}}{1000} n$.
	
	Let $\mathbf{C}_{1}$ be the set of content types requested in Phase 1. Let $p = \sum_{c \in \mathbf{C}_{1}} p_c $ be the total mass of all content types $c \in \mathbf{C}_{1}$. Let $\hat{p}_c$ be the fraction of requests for content-type $c$ in Phase 1. By the definition of $\mathbf{C}_{1}$, the total empirical mass of all content types $c \in \mathbf{C}_{1}$ is obviously $\hat{p} = \sum_{c \in \mathbf{C}_{1}} \hat{p}_c  = 1$.
	
	Recall that there are $n^{\gamma}$ arrivals in Phase 1. Let $r = n^{\gamma}$. We now use the Chernoff bound to compute a lower bound on the true mass $p$, using a technique similar to that used in \cite{MS00} (Lemma 4). By the Chernoff bound, we know that,
	\begin{eqnarray*}
		\mathbb{P}(\hat{p} > (1+\kappa) p) \leq \exp \bigg(-\frac{pr\kappa^2}{3} \bigg).
	\end{eqnarray*}
	Let $\delta = \exp \bigg(-\dfrac{pr\kappa^2}{3} \bigg)$, then, we have that, with probability greater than $ 1- \delta$,
	\begin{eqnarray*}
		\hat{p} - p > \sqrt{\dfrac{-3 p \log \delta}{r}}.
	\end{eqnarray*}
	Solving for $p$, we get that, with probability greater than $ 1- \delta$,
	$p > 1 - \dfrac{3 \log (1/\delta)}{2r},$
	for $n$ large enough.
	Let $\delta = 1/n$, then we have that, with probability greater than $ 1 - 1/n$,
	$p > 1 - \dfrac{3 \log n}{2 n^{\gamma}}.$
	Conditioned on the event $E_4$, there are at least $\frac{T(n) \bar{\lambda} n -  n^{\gamma}}{2} $ arrivals in Phase 2. The remainder of this proof is conditioned on $E_4$. Let $A_2$ be the number of arrivals of types $c \in \mathbf{C}_{1}$ in phase 2. Let $E_7$ be the event that
	$$A_2 > \dfrac{T(n) \bar{\lambda} n -  n^{\gamma}}{2}\bigg(1 - \dfrac{3 \log n}{2 n^{\gamma}} \bigg).$$ Since the expected number of arrivals of content types $c \in \mathbf{C}_{1}$ in Phase 2 is at least
	$$(T(n) \bar{\lambda} n - n^{\gamma})\bigg(1 - \dfrac{3 \log n}{2 n^{\gamma}} \bigg), $$
	using the Chernoff bound, we can show that
	$
	\mathbb{P}(E_7^c) = \oo(1/n).
	$
	The rest of this proof is conditioned on $E_7$.
	By our assumption, the number of servers which can serve arrivals of types $c \in \mathbf{C}_{1}$ in Phase 2 is $f(n)$. Therefore, if at least $A_2/2$ requests are to be served in Phase 2, the sum of the service times of these $A_2/2$ requests should be less than $T(n) f(n)$ (since the number of servers which can serve these requests is $f(n)$). Let $E_8$ be the event that the sum of $A_2/2$ independent Exponential random variables with mean 1 is less than $ T(n) f(n)$. By substituting $v = A_2/2$ and $a = T(n) f(n)$ in Lemma \ref{lemma:sum_of_exponentials}, we have that,
	\begin{eqnarray*}
		\mathbb{P}(E_8) &\leq& \exp\bigg(\frac{A_2}{2}-T(n)\bigg) \bigg(\dfrac{2T(n)f(n)}{A_2} \bigg)^{\frac{A_2}{2}} \\
		&\leq& \exp\bigg(\frac{A_2}{2}\bigg) \bigg(\dfrac{2T(n)f(n)}{A_2} \bigg)^{\frac{A_2}{2}} = \oo\bigg(\frac{1}{n}\bigg)
	\end{eqnarray*}
	for $n$ large enough. Hence,
	\begin{eqnarray*}
		\mathbb{P}\bigg(D_2 \geq \frac{A_2}{2}\bigg) &\geq& 1 - \mathbb{P}(E_1^c) - o\bigg(\frac{1}{n}\bigg) \\
		\Rightarrow E[D_2] &=& \Omega\big(T(n)n^{1-\gamma}n^{\gamma/\beta}\big).
	\end{eqnarray*}
\end{proof}

\begin{proof} (Proof of Theorem \ref{thm:converse}) \\
	We consider two cases: \\
	Case I: The learning phase lasts for the first $n^\gamma$ arrivals where $0 \leq \gamma < 1$. \\
	Let $D_1$ be the number of requests deferred in Phase 1 and $D$ be total number of requests deferred in the interval of interest. Then, we have that,
	\begin{eqnarray*}
		E[D] = E[D_1] + E[D_2].
	\end{eqnarray*}
	By Lemmas \ref{lemma:E_2} and \ref{lemma:phase2_drops} and since $T(n) = \Omega(1)$, we have that,
	\begin{eqnarray*}
		E[D] &\geq& n^\gamma - ({n^{\gamma}\log n})^{\frac{1}{\beta-1}}\log n + E[D_2] \\
		&=& \Omega(nT(n))^{\frac{1}{2-1/\beta}}.
	\end{eqnarray*}
	Case II: The learning phase lasts for longer than the time taken for the first $n$ arrivals.\\
	By Lemma \ref{lemma:E_2}, the number of requests deferred in the first $n$ arrivals is at least $n - \OO(n^{1/\beta}\log n)$ with probability greater than $ 1-1/\log n$. Therefore, we have that,
	\begin{eqnarray*}
		E[D] &\geq& \bigg(n - \OO(n^{1/\beta}\log n)\bigg) \bigg(1-\frac{1}{\log n}\bigg) = \Omega(n) \\
		&=& \Omega(nT(n))^{\frac{1}{2-1/\beta}}.
	\end{eqnarray*}
\end{proof}

\subsection{Proof of Theorem 2}
In this section, we provide an outline of the proof of Theorem \ref{thm:converse_2}. The proof follows on the same lines as the proof of Theorem \ref{thm:converse}.
\begin{enumerate}
	\item First, we show that w.h.p., among the first $n^{\gamma}$ arrivals, i.e., during the learning phase, $\Omega(n^{\gamma})$ requests are deferred (Lemma \ref{lemma:E_2}).
	\item Since we are studying the performance of the MYOPIC policy for the Continuous Change Model, the relative order of popularity of contents keeps changing in the interval of interest. If the learning phase lasts for the first $n^{\gamma}$ arrivals for some $0<\gamma \leq 1$, we show that under Assumption \ref{ass:zipf}, w.h.p., in the learning phase, only $\OO(n^{\gamma/\beta})$ content types are requested.
	\item Next, we show that the expected the number of requests in Phase 2 for content types not requested in Phase 1 is $\Omega(n^{1-\gamma}n^{\gamma/\beta})$. Using this, we compute a lower bound on the number of requests deferred in Phase 2 (after the learning phase) by any learning-based static storage policy. This results follows by the same arguments as the proof of Lemma \ref{lemma:phase2_drops}.
	\item Using Steps 1 and 3, we lower bound the number of requests deferred in the interval of interest.
\end{enumerate}

\subsection{Proof of Theorem 3}
\label{sec:proof3}

\noindent We first present an outline the proof of Theorem \ref{thm:MYOPIC_static_arrival_rates}.

\begin{enumerate}
	\item We first show that under Assumption \ref{ass:zipf}, on every arrival in the interval of interest ($T(n)$), there are $\Theta(n)$ idle servers w.h.p. (Lemma \ref{lemma:F_1}).
	\item Next, we show that w.h.p., in the interval of interest of length $T(n)$, only $\OO\big((nT(n)\big)^{\frac{1}{\beta}})$ unique content types are requested (Lemma \ref{lemma:types_of_arrivals}).
	\item Conditioned on Steps 1 and 2, we show that, the MYOPIC policy ensures that in the interval of interest, once a content type is requested for the first time, there is always at least one idle server which can serve an incoming request for that content.
	\item Using Step 3, we conclude that, in the interval of interest, only the first request for a particular content type will be deferred. The proof of Theorem \ref{thm:MYOPIC_static_arrival_rates} then follows from Step 2.
\end{enumerate}

\begin{lemma}
	\label{lemma:occupancy}
	Let the cumulative arrival process to the content delivery system be a Poisson process with rate $\bar{\lambda} n$. At time $t$, let $\chi(t)$ be the number of occupied servers under the MYOPIC storage policy. Then, we have that, $\chi(t) \leq_{st} S(t)$, where $S(t)$ is a poisson random variable with rate $\bar{\lambda} n (1-e^{-t})$.
\end{lemma}
\begin{proof}
	Consider an $M/M/\infty$ queue where the arrival process is Poisson($\bar{\lambda} n$). Let $S(t)$ be the number of occupied servers at time $t$ in this system. It is well known that $S(t)$ is a Poisson random variable with rate $\bar{\lambda} n (1-e^{-t})$. Here we provide a proof of this result for completeness. Consider a request $r^*$ which arrived into the system at time $t_0 < t$. If the request is still being served by a server, we have that,
	\begin{eqnarray*}
		t_0 + \mu(r^*) > t,
	\end{eqnarray*}
	where $\mu(r^*)$ is the service time of request $r^*$. Since $\mu(r^*) \sim$ Exp(1), we have that,
	\begin{eqnarray*}
		\mathbb{P}(\mu(r^*) > t - t_0 | t_0) = e^{-(t - t_0)}.
	\end{eqnarray*}
	Therefore,
	\begin{eqnarray*}
		\mathbb{P}(r^* \text{ in the system at time $t$}) &\leq& \int_{0}^t \frac{1}{t} e^{-(t - t_0)} dt_0 \\
		&=& \frac{1-e^{-t}}{t}.
	\end{eqnarray*}
	Therefore, every request that arrived in the system is still in the system with probability at most $\frac{1-e^{-t}}{t}$. Since the arrival process is Poisson, the number of requests in the system at time $t$ is stochastically dominated by a Poisson random variable with rate $\bar{\lambda} n t \big(\frac{1-e^{-t}}{t}\big) = \bar{\lambda} n (1-e^{-t})$.
	
	To show $\chi(t) \leq_{st} S(t)$, we use a coupled construction similar to Figure \ref{fig:coupling}. The intuition behind the proof is the following: the rate of arrivals to the content delivery system and the $M/M/\infty$ system (where each server can serve all types of requests) is the same. The content delivery system serves fewer requests than the $M/M/\infty$ system because some requests are deferred even when the servers are idle. Hence, the number of busy servers is the content delivery system is stochastically dominated by the number of busy servers in the $M/M/\infty$ queueing system.
\end{proof}

\begin{lemma}
	\label{lemma:F_1}
	Let the interval of interest be $[t_0, t_0 + T(n)]$ where $T(n) = \oo(n^{\beta-1})$ and $\varepsilon \leq \frac{1-\bar{\lambda}}{2}$. Let $F_1$ be the event that at the instant of each arrival in the interval of interest, the number of idle servers in the system is at least $\big(1-\bar{\lambda}-\varepsilon \big) n$. Then,
	$
	\mathbb{P}(F_1^c) = \oo\big(\frac{1}{n}\big).
	$
\end{lemma}
\begin{proof}
	Let $F_2$ be the event that the number of arrivals in $[t_0, t_0 + T(n)] \leq nT(n)(\bar{\lambda}+\varepsilon)$. Using the Chernoff bound for the  Poisson process, we have that, \begin{eqnarray*}
		\mathbb{P}(F_2^c) = o\bigg(\frac{1}{n}\bigg).
	\end{eqnarray*}
	Consider any $t \in [t_0, t_0 + T(n)]$. By Lemma \ref{lemma:occupancy}, $\chi(t) \leq_{st} S(t) $, where $S(t) \sim$ Poisson($\bar{\lambda} n (1-e^{-t})$). Therefore,
	\begin{eqnarray*}
		\mathbb{P}(\chi(t)>(\bar{\lambda}+\varepsilon)n) \leq \mathbb{P}(S(t)>(\bar{\lambda}+\varepsilon)n).
	\end{eqnarray*}
	Moreover, $S(t) \leq_{st} W(t)$ where $W(t) =$ Poisson($\bar{\lambda} n$). Therefore, using the Chernoff bound for $W(t)$, we have that,
	\begin{eqnarray*}
		\mathbb{P}(S(t)>(\bar{\lambda}+\varepsilon)n) \leq \mathbb{P}(W(t)>(\bar{\lambda}+\varepsilon)n) = e^{-c_1n},
	\end{eqnarray*}
	for some constant $c_1 > 0$. Therefore,
	\begin{eqnarray*}
		\mathbb{P}(F_1^c) &\leq& \mathbb{P}(F_2^c) + (\bar{\lambda}+\varepsilon)nT(n) \mathbb{P}(\chi(t)>(\bar{\lambda}+\varepsilon)n)\\
		&=& o\bigg(\frac{1}{n}\bigg).
	\end{eqnarray*}
\end{proof}

\begin{lemma}
	\label{lemma:types_of_arrivals}
	Let $F_3$ be the event that in the interval of interest of duration $T(n)$ such that $T(n) = \oo(n^{\beta-1})$, no more than $\OO((nT(n))^{1/\beta})$ different types of contents are requested. Then,
	$
	\mathbb{P}(F_3^c) = \oo\big(\frac{1}{n}\big).
	$
\end{lemma}
\begin{proof}
	Recall from the proof of Lemma \ref{lemma:E_1} that the total mass of all content types $k, .. m=\alpha n$ is
	\begin{eqnarray*}
		\sum_{i=k}^{\alpha n} p_i &\leq& \frac{1}{0.9} \dfrac{1}{(k-1)^{\beta-1}}.
	\end{eqnarray*}
	Now, for $k = (nT(n))^{1/\beta}$ + 1, we have that,
	\begin{eqnarray*}
		\sum_{i=k}^{\alpha n} p_i  \leq \frac{1}{0.9} (nT(n))^{-\frac{\beta-1}{\beta}}.
	\end{eqnarray*}
	Conditioned on the event $F_2$ defined in Lemma \ref{lemma:F_1}, the expected number of requests for content types $k, k+1, .. \alpha n$ is less than $\frac{1}{0.9} (\bar{\lambda} + \varepsilon) (nT(n))^{1/\beta}$. Using the Chernoff bound, the probability that there are more than $\frac{2}{0.9} (\bar{\lambda} + \varepsilon) (nT(n))^{1/\beta}$ requests for content types $k, k+1, .. \alpha n$ in the interval of interest is less than $ \frac{1}{n^2}$ for $n$ large enough.
	
	Therefore, with probability greater than $ 1-1/n^2 - \mathbb{P}(F_2^c)$, the number different types of contents requests for in the interval of interest is less than $ (nT(n))^{1/\beta} + \frac{2}{0.9} (\bar{\lambda} + \varepsilon) (nT(n))^{1/\beta}$. Hence the result follows.
\end{proof}

\begin{proof} (Proof of Theorem \ref{thm:MYOPIC_static_arrival_rates})\\
	Let $F_4$ be the event that, in the interval of interest, every request for a particular content type except the first request is not deferred.
	The rest of this proof is conditioned on $F_1$ and $F_3$. Let $U(t)$ be the number of unique contents which have been requested in the interval of interest before time $t$ for $t \in [t_0, t_0 + T(n) ]$. Conditioned on $F_3$, as defined in Lemma~\ref{lemma:types_of_arrivals}, $U(t) \leq k_1 (nT(n))^{1/\beta}$ for some constant $k_1>0$ and $n$ large enough. Conditioned on $F_1$, there are always $(1- \bar{\lambda} - \varepsilon) n$ idle servers in the interval of interest. \\
	\newline CLAIM: For every $i$ and $n$ large enough, once a content $C_i$ is requested for the first time in the interval of interest, the MYOPIC policy ensures that there is always at least 1 idle server which can serve a request for $C_i$.\\
	\newline Note that since $T(n) = \oo(n^{\beta-1})$, $(nT(n))^{1/\beta} = \oo(n)$. Let $n$ be large enough such that $k_1 (nT(n))^{1/\beta} < (1- \bar{\lambda} - \varepsilon) n$, i.e., at any time $t \in [t_0, t_0 + T(n)]$, the number of idle servers is greater than $ U(t)$. We prove the claim by induction. Let the claim hold for time $t^-$ and let there be a request at time $t$ for content $C_i$. If this is not the first request for $C_i$ in $[t_0, t_0 + T(n)]$, by the claim, at $t = t^-$, there is at least one idle server which can serve this request. In addition, if there is exactly one server which can serve $C_i$ at $t^-$, then the MYOPIC policy replaces the content of some other idle server with $C_i$. Since there are more than $k_1 (nT(n))^{1/\beta}$ idle servers and $U(t) < k_1 (nT(n))^{1/\beta}$, at $t^+$, each content type requested in the interval of interest so far, is stored on at least one currently idle server. Therefore, conditioned on $F_1$ and $F_3$, every request for a particular content type except the first request, is not deferred. \\
	\newline Hence, putting everything together,
	\begin{eqnarray*}
		\mathbb{P}(F_4) \geq 1 - \mathbb{P}(F_1^c) - \mathbb{P}(F_3^c),
	\end{eqnarray*}
	thus $\mathbb{P}(F_4) \rightarrow 1$ as $n \rightarrow \infty$ and the result follows.
\end{proof}


\subsection{Proof of Theorem \ref{thm:MYOPIC_changing_arrival_rates}}
\label{sec:proof4}

\noindent We first present an outline of the proof of Theorem \ref{thm:MYOPIC_changing_arrival_rates}. 

\begin{enumerate}
	\item Since we are studying the performance of the MYOPIC policy for the Continuous Change Model, the relative order of popularity of contents keeps changing in the interval of interest. We show that w.h.p., the number of content types which are in the $n^{1/\beta}$ most popular content types at least once in the interval of interest is $\OO(n^{1/\beta})$ (Lemma \ref{lemma:change_in_heavy_hitters}).
	\item Next, we show that w.h.p., in the interval of interest of length $b$, only $\OO(n^{1/\beta})$ content types are requested (Lemma \ref{lemma:types_of_arrivals_changing_arrival_rates}).
	\item By Lemma \ref{lemma:F_1} and the proof of Theorem \ref{thm:MYOPIC_static_arrival_rates}, we know that, conditioned on Step 3, the MYOPIC storage policy ensures that in the interval of interest, once a content type is requested for the first time, there is always at least one idle server which can serve an incoming request for that content. Using this, we conclude that, in the interval of interest, only the first request for a particular content type will be deferred. The proof of Theorem \ref{thm:MYOPIC_changing_arrival_rates} then follows from Step 2.
\end{enumerate}


\begin{lemma}
	\label{lemma:change_in_heavy_hitters}
	Let $G_1$ be the event that, in the interval of interest of length $b$, the number of times that a content among the current top $n^{1/\beta}$ most popular contents changes its position in the popularity ranking is at most $\frac{4b}{\alpha} n^{1/\beta} \nu$. Then, $
	P(G_1) \geq 1-o\big( \frac{1}{n} \big).$
\end{lemma}
\begin{proof}
	The expected number of clock ticks in $b$ time-units is $bn\nu$. The probability that a change in arrival process involves at least one of the current $n^{1/\beta}$ most popular contents is $ \frac{n^{1/\beta}}{\alpha n}$. Therefore, the expected number of changes in arrival process which involve at least one of the current $n^{1/\beta}$ most popular contents is  $\frac{2b\nu}{\alpha} n^{1/\beta}$. By the Chernoff bound, we have that
	$
	P(G_1) \geq 1-\oo\big( \frac{1}{n} \big).
	$
\end{proof}

\begin{lemma}
	\label{lemma:types_of_arrivals_changing_arrival_rates}
	Let $G_2$ be the event that in the interval of interest, no more than $\OO(n^{1/\beta})$ different types of contents are requested. Then,
	$
	\mathbb{P}(G_2^c) = \oo\big(\frac{1}{n}\big).
	$
\end{lemma}
\begin{proof}
	Conditioned on the event $G_1$ defined in Lemma \ref{lemma:change_in_heavy_hitters}, we have that in the interval of interest, at most $\big(\frac{2b}{\alpha} \nu + 1\big)n^{1/\beta}$ different contents are among the top $n^{1/\beta}$ most popular contents. Given this, the proof follows the same lines of arguments as in the proof of Lemma \ref{lemma:types_of_arrivals}.
\end{proof}
The proof of the theorem then follows from Lemma \ref{lemma:types_of_arrivals_changing_arrival_rates} and uses the same line of arguments as in the proof of Theorem \ref{thm:MYOPIC_static_arrival_rates}.

\subsection{Proof of Theorem 5}
\label{sec:proof1}

To show that GENIE is the optimal policy, we consider the process $X(t)$ which is the number of occupied servers at time $t$ when the storage policy is GENIE. Let $Y(t)$ be the number of occupied servers at time $t$ for some other storage policy $\mathcal{A} \in \mathds{A}$. We construct a coupled process $(X^*(t),Y^*(t))$ such that the marginal rates of change in $X^*(t)$ and $Y^*(t)$ is the same as that of $X(t)$ and $Y(t)$ respectively.

Recall $\bar{\lambda} = \dfrac{\sum_{i=1}^m \lambda_i}{n}$. At time $t$, let $\mathbf{C^{GENIE}}(t)$ and $\mathbf{C^{\mathcal{A}}}(t)$ be the sets of contents stored on idle servers by GENIE and $\mathcal{A}$ respectively. The construction of the coupled process $(X^*(t),Y^*(t))$ is described in Figure \ref{fig:coupling}. We assume that the system starts at time $t=0$ and $X^*(0) = Y^*(0) = 0$. In this construction, we maintain two counters $Z_{X^*}$ and $Z_{Y^*}$ which keep track of the number of departures from the system. Let $Z_{X^*}(0) = Z_{Y^*}(0) = 0$. Let Exp$(\mu)$ be an Exponential random variable with mean $\frac{1}{\mu}$ and Ber$(p)$ be a Bernoulli random variable which is 1 with probability (w.p.) $p$.

\begin{figure}[h]
	\hrule
	\vspace{0.1in}
	\begin{algorithmic}[1]
		\STATE Generate: ARR $\sim$ Exp$(n \bar{\lambda})$, DEP $\sim$ Exp$(\max\{X^*,Y^*\})$
		\STATE $t=t+\min\{$ARR,DEP$\}$
		\IF {ARR$<$DEP, }
		\IF {($X^*=Y^*$)}
		\STATE Generate $u_1 \sim$ Ber$\bigg(\dfrac{\sum_{i \in \mathbf{C^{GENIE}}(t)} \lambda_i}{n \bar{\lambda}}\bigg)$
		\IF {($u_1 =1$)}
		\STATE $X^* \gets X^* + 1$
		\STATE Generate $u_2 \sim$ Ber$\bigg(\dfrac{\sum_{i \in \mathbf{C^{\mathcal{A}}}(t)} \lambda_i}{\sum_{i \in \mathbf{C^{GENIE}(t)}}\lambda_i}\bigg)$
		\STATE \textbf{if} ($u_2 =1$) \textbf{then} $Y^* \gets Y^* +1$
		\ENDIF
		\ELSE
		\STATE Generate $u_1 \sim$ Ber$\bigg(\dfrac{\sum_{i \in \mathbf{C^{GENIE}}(t)} \lambda_i}{n \bar{\lambda}}\bigg)$
		\STATE \textbf{if}($u_1 =1$) \textbf{then} $X^* \gets X^* + 1$
		\STATE Generate $u_2 \sim$ Ber$\bigg(\dfrac{\sum_{i \in \mathbf{C^{\mathcal{A}}}(t)} \lambda_i}{\sum_{i \in \mathbf{C^{GENIE}(t)}}\lambda_i}\bigg)$
		\STATE \textbf{if}($u_2 =1$) \textbf{then} $Y^* \gets Y^* + 1$
		\ENDIF
		\ELSE
		\IF{($X^* \geq Y^*$)}
		\STATE $X^* \gets X^* - 1$, $Z_{X^*} \gets Z_{X^*} + 1$
		\STATE Generate $u_3 \sim$ Ber$\bigg(\dfrac{Y^*}{X^*}\bigg)$
		\STATE \textbf{if} ($u_3 =1$) \textbf{then} $Y^* \gets Y^* - 1$, $Z_{Y^*} \gets Z_{Y^*} + 1$
		\ELSE
		\STATE $Y^* \gets Y^* - 1$, $Z_{Y^*} \gets Z_{Y^*} + 1$
		\STATE Generate $u_4 \sim$ Ber$\bigg(\dfrac{X^*}{Y^*}\bigg)$
		\STATE \textbf{if} ($u_4 =1$) \textbf{then} $X^* \gets X^* - 1$, $Z_{X^*} \gets Z_{X^*} + 1$
		\ENDIF
		\ENDIF
		\STATE Goto 1
	\end{algorithmic}
	\vspace{0.1in}
	\hrule
	\caption{Coupled Process}
	\label{fig:coupling}
\end{figure}

\begin{lemma}
	\label{lemma:marginals}
	$X^*(t)$ and $Y^*(t)$ have the same marginal rates of transition as $X(t)$ and $Y(t)$ respectively.
\end{lemma}
\begin{proof}
	Consider a small interval of time $[t_0, t_0 + \delta]$. By the definition of $X(t)$,
	\begin{eqnarray*}
		\mathbb{P}(X(t_0+\delta) = X(t_0)+1)  &\approx&  \bigg(\sum_{i \in \mathbf{C^{GENIE}}(t)} \lambda_i\bigg) \delta, \\
		\mathbb{P}(X(t_0+\delta) = X(t_0 )-1) &\approx& X(t_0) \delta.
	\end{eqnarray*}
	The above probabilities are implicitly conditioned on a suitable state definition for the system; we henceforth drop the conditioning on the state for notational compactness. For the process $X^*(t)$,
	\begin{eqnarray*}
		\mathbb{P}(X^*(t_0+\delta) = X^*(t_0 )+1) &\approx& n \bar{\lambda} \bigg(\dfrac{\sum_{i \in \mathbf{C^{GENIE}}(t)} \lambda_i}{n \bar{\lambda}}\bigg) \delta \\
		& = & \bigg(\sum_{i \in \mathbf{C^{GENIE}}(t)} \lambda_i\bigg) \delta.
	\end{eqnarray*}
	If $(X^*(t_0) \geq Y^*(t_0))$,
	\begin{eqnarray*}
		\mathbb{P}(X^*(t_0+\delta) = X^*(t_0 )-1) \approx X^*(t_0) \delta,
	\end{eqnarray*}
	and if $(X^*(t_0) < Y^*(t_0))$,
	\begin{eqnarray*}
		\mathbb{P}(X^*(t_0+\delta) = X^*(t_0 )-1) &\approx& Y^*(t_0) \dfrac{X^*(t_0)}{Y^*(t_0)} \delta \\
		&=& X^*(t_0) \delta.
	\end{eqnarray*}
	The approximations become exact as $\delta \rightarrow 0$, since the inter-event (arrival or departure) times are exponential. This proves the lemma for $X^*$ and $X$.
	\newline \newline By the definition of $Y(t)$,
	\begin{eqnarray*}
		\mathbb{P}(Y(t_0+\delta) = Y(t_0 )+1)  &\approx&  \bigg(\sum_{i \in \mathbf{C^{\mathcal{A}}}(t)} \lambda_i\bigg) \delta, \\
		\mathbb{P}(Y(t_0+\delta) = Y(t_0 )-1) &\approx& Y(t_0) \delta.
	\end{eqnarray*}
	Consider the case when $Y^*(t_0 ) = X^*(t_0 )$.
	\newline From Section \ref{subsec:optimal}, we know that, under the GENIE storage policy, if the number of idle servers at time $t$ is $\idle (t)$, they store the $\idle (t)$ most popular contents. Given this, if $X^*(t_0 ) = Y^*(t_0 )$, $\dfrac{\sum_{i \in \mathbf{C^{\mathcal{A}}}(t)} \lambda_i}{\sum_{i \in \mathbf{C^{GENIE}(t)}}\lambda_i} \leq 1$. Therefore, $u_2$ as defined in Step 8 of the coupling construction is a valid bernoulli random variable and in addition, $u_1 \times u_2$ is a bernoulli random variable with parameter $\bigg(\dfrac{\sum_{i \in \mathbf{C^{\mathcal{A}}}(t)} \lambda_i}{n \bar{\lambda}}\bigg)$. Therefore, we have that,
	\begin{eqnarray*}
		\mathbb{P}(Y^*(t_0+\delta) = Y^*(t_0 )+1) &\approx& n \bar{\lambda} \bigg(\dfrac{\sum_{i \in \mathbf{C^{\mathcal{A}}}(t)} \lambda_i}{n \bar{\lambda}}\bigg) \delta \\
		& = & \bigg(\sum_{i \in \mathbf{C^{\mathcal{A}}}(t)} \lambda_i\bigg) \delta.
	\end{eqnarray*}
	If $Y^*(t_0 ) \neq X^*(t_0 )$,
	\begin{eqnarray*}
		\mathbb{P}(Y^*(t_0+\delta) = Y^*(t_0 )+1) &\approx& n \bar{\lambda} \bigg(\dfrac{\sum_{i \in \mathbf{C^{\mathcal{A}}}(t)} \lambda_i}{n \bar{\lambda}}\bigg) \delta \\
		& = & \bigg(\sum_{i \in \mathbf{C^{\mathcal{A}}}(t)} \lambda_i\bigg) \delta.
	\end{eqnarray*}
	If $(X^*(t_0) \geq Y^*(t_0))$,
	\begin{eqnarray*}
		\mathbb{P}(Y^*(t_0+\delta) = Y^*(t_0 )-1) &\approx& X^*(t_0) \dfrac{Y^*(t_0)}{X^*(t_0)} \delta\\
		&=& Y^*(t_0) \delta,
	\end{eqnarray*}
	and if $(X^*(t_0) < Y^*(t_0))$,
	\begin{eqnarray*}
		\mathbb{P}(Y^*(t_0+\delta) = Y^*(t_0 )-1) &\approx& Y^*(t_0) \delta \\
		&=& Y^*(t_0) \delta.
	\end{eqnarray*}
	This completes the proof.
\end{proof}

\begin{lemma}
	\label{lemma:marginals_2}
	Let $D^{(GENIE)}(t)$ be the number of jobs deferred by time $t$ by the GENIE adaptive storage policy and $D^{(\mathcal{A})}(t)$ to be the number of jobs deferred by time $t$ by a policy $\mathcal{A} \in \mathds{A}$. In the coupled construction, let $W^*(t)$ be the number of arrivals by time $t$. Let, $D^{X^*}(t) = W^*(t) - Z_{(X^*)}(t) - X^*(t)$ and $D^{Y^*}(t)=W^*(t) - Z_{(Y^*)}(t) - Y^*(t)$. Then, $D^{X^*}(t)$ and $D^{Y^*}(t)$ have the same marginal rates of transition as $D^{(GENIE)}(t)$ and $D^{(\mathcal{A})}(t)$ respectively.
\end{lemma}
\begin{proof}
	This follows from Lemma \ref{lemma:marginals} due to the fact that $X(t)$ have the same distribution as $X^*(t)$ and the marginal rate of increase of $D^{X^*}(t)$ given $X^*(t)$ is the same as the rate of increase of $D^{(GENIE)}(t)$ given $X(t)$. The result for $D^{Y^*}(t)$ follows by the same argument.
\end{proof}

\begin{lemma}
	\label{lemma:x_stochastic_dominance_y}
	$X^* \geq Y^*$ for all $t$ on every sample path.
\end{lemma}
\begin{proof} 
	The proof follows by induction. $X^*(0) = Y^*(0)$ by construction.
	Let $X^*(t_0^-) \geq Y^*(t_0^-)$ and let there be an arrival or departure at time $t_0$. There are 4 possible cases:
	\begin{itemize}
		\item[i:] If ARR$<$DEP and $X^*(t_0^-) = Y^*(t_0^-)$, $Y^*(t_0) = Y^*(t_0^-) + 1$ only if $X^*(t_0) = X^*(t_0^-) + 1$. Therefore, $X^*(t_0) \geq Y^*(t_0)$.
		\item[ii:] If ARR$<$DEP and $X^*(t_0^-) > Y^*(t_0^-)$, $Y^*(t_0) \leq Y^*(t_0^-) + 1 \leq X^*(t_0^-) \leq X^*(t_0)$. Therefore, $X^*(t_0) \geq Y^*(t_0)$.
		\item[iii:] If DEP$<$ARR and $X^*(t_0^-) = Y^*(t_0^-)$, $X^*(t_0) = Y^*(t_0)$.
		\item[iv:] If DEP$<$ARR and $X^*(t_0^-) > Y^*(t_0^-)$, $X^*(t_0) = X^*(t_0^-)-1 \geq Y^*(t_0^-) \geq Y^*(t_0)$. Therefore, $X^*(t_0) \geq Y^*(t_0)$.
	\end{itemize}
\end{proof}

\begin{lemma}
	\label{lemma:departures_stochastic_dominance}
	$Z_{X^*} \geq Z_{Y^*}$ for all $t$ on every sample path.
\end{lemma}
\begin{proof}
	The proof follows by induction. Since the system starts at time $t=0$, $Z_{X^*}(0) = Z_{Y^*}(0)$. Let $Z_{X^*}(t_0^-) \geq Z_{Y^*}(t_0^-)$ and let there be a departure at time $t_0$. By Lemma \ref{lemma:x_stochastic_dominance_y}, we know that, $X^*(t_0^-) \geq Y^*(t_0^-)$. Therefore, $Z_{X^*}(t_0) \geq Z_{Y^*}(t_0)$ by the coupling construction.
\end{proof}

\begin{proof} (Proof of Theorem \ref{thm:stochastic_dominance}) \\
	By Lemmas \ref{lemma:x_stochastic_dominance_y} and \ref{lemma:departures_stochastic_dominance}, for any sample path,
	\begin{eqnarray*}
		X^*(t) + Z_{X^*}(t) \geq Y^*(t) + Z_{Y^*}(t).
	\end{eqnarray*}
	Therefore, for every sample path, the number of requests already served (not deferred) or being served by the servers by a content delivery system implementing the GENIE policy is more than that by any other storage policy. This implies that for each sample path, the number of requests deferred by GENIE is less than that of any other storage policy. Sample path dominance in the coupled system implies stochastic dominance of the original process. Using this and Lemma \ref{lemma:marginals_2}, we have that,
	\begin{eqnarray*}
		D^{(GENIE)}(t) \leq_{st} D^{(\mathcal{A})}(t).
	\end{eqnarray*}
\end{proof}

\subsection{Proof of Theorem \ref{thm:GENIE_static_arrival_rates}}
\begin{proof}
	The key idea of the GENIE policy is to ensure that at any time $t$, if the number of idle servers is $\idle (t)$, the $\idle (t)$ most popular contents are stored on exactly one idle server each. Since the total number of servers is $n$, and the number of content-types is $m = \alpha n$ for some constant $\alpha > 1$, all content-types $C_i$ for $i > n$ are never stored on idle servers by the GENIE policy. This means that under the GENIE policy, all arrivals for content types $C_i$ for $i > n$ are deferred. For $\beta > 1$,for all $i$,$p_i \geq \frac{1}{c_z} i^{-\beta}$, for some constant $c_z < \infty$.
	The cumulative mass of all content types $i = n+1, .. \alpha n$ is
	\begin{eqnarray*}
		\sum_{i=n+1}^{\alpha n} p_i &\geq& \sum_{i=k}^{\alpha n} \frac{1}{c_z } i^{-\beta} \geq \int_{n+1}^{\alpha n+1} \frac{1}{c_z } i^{-\beta} di \\
		&\geq& \frac{0.9}{c_z  (\beta -1)} \dfrac{1}{(n+1)^{\beta-1}},
	\end{eqnarray*}
	for $n$ large enough.
	
	Let the length of the interval of interest be $b$. The expected number of arrivals of types $n+1, n+2, .. \alpha n$, in the interval of interest is at least $\dfrac{0.9 b \bar{\lambda} n}{c_z (\beta -1)} \\ \dfrac{1}{(n+1)^{\beta-1}}$. Therefore, the expected number of jobs deferred by the GENIE policy in an interval of length $b$ is $\Omega (n^{2-\beta})$.\\
\end{proof}

\subsection{Proof of Theorem \ref{thm:adaptation_cost}}
\begin{proof}
	From the proof of Theorem \ref{thm:MYOPIC_static_arrival_rates}, we know that if $T = o(n^{\beta-1})$, w.h.p.,
	\begin{itemize}
		\item[-] no more than $O(nT)^{1/\beta}$ different types of contents are requested,
		\item[-] once a content $C_i$ is requested for the first time, the MYOPIC policy ensures that there is always at least 1 idle server which can serve a request for $C_i$.
	\end{itemize}
	
	It follows that once a content is requested for the first time, there is at least one copy of that content in the system (more specifically, there is at least one copy of that content on an idle server). Therefore, w.h.p., the number of external fetches is equal to the number of unique content types requested in the interval of interest and the result follows.
	
	For the GENIE policy, before the first arrival, the GENIE policy fetches the $n$ most popular contents to place on the servers.
	
	Let the number of idle servers at $t^-$ be $k(t)$ and let there be a departure from the system at time $t$. After this departure, the content of the new idle server is replaced with $C_{k(t^-)+1}$. From Lemma \ref{lemma:F_1}, we have that with probability $\geq 1 -\oo\big(\frac{1}{n}\big)$, $\Theta(n)$ servers are idle at all times in the interval of interest. Therefore, $k(t^-)+1 > \epsilon n$ for some $\epsilon > 0$ and $\lambda_{k(t^-)+1} \leq \frac{\bar{\lambda} n}{(\epsilon n)^{\beta}}$. The number of currently busy servers serving a request for content $k(t^-)+1$ is stochastically dominated by a Poisson random variable with rate $\frac{\bar{\lambda} n}{(\epsilon n)^{\beta}}$. Therefore, at time $t^+$, with probability $\geq 1 - \frac{\bar{\lambda} n}{(\epsilon n)^{\beta}}$, there is no currently busy server in the system serving a request for $C_{k(t^-)+1}$. By the properties of the GENIE policy, the other $k(t^-)$ idle servers store the $k(t^-)$ most popular contents. Therefore, content $k(t^-)+1$ is not available in the system (on a busy or idle server) at time $t^+$ and will be fetched from the back-end server. Therefore, w.h.p., each departure is followed by an external fetch. Since there are $\Theta(nT)$ departures in an interval of duration $T$, the result follows.
	
\end{proof}

\section{Related Work}
\label{sec:related}
Our model of content delivery systems shares several features with
recent models and analyses for content placement and request
scheduling in multi-server queueing systems
\cite{LLM12,XT13,LLM13,Whitt07}. All these works either assume known
demand statistics, or a low-dimensional regime (thus permiting
``easy'' learning). Our study is different in its focus on unknown,
high-dimensional and time-varying demand statistics, thus making it
difficult to consistently estimate statistics.  Our setting also
shares some aspects of estimating large alphabet distributions with
only limited samples, with early contributions from Good and
Turing~\cite{Good53}, to recent variants of such
estimators~\cite{MS00,VV11}.

Our work is also related to the rich body of work on the content
replication strategies in peer-to-peer networks, e.g., \cite{TM13,
	KRR02, LC02, KRT07, WL12, ZFC13, CMGLT12, cluster02}.  Replication
is used in various contexts: \cite{TM13} utilizes it in a setting with
large storage limits, \cite{KRR02,LC02} use it to decrease the time
taken to locate specific content, and \cite{ZFC13,CMGLT12,cluster02}
use it to increase bandwidth in the setting of video
streaming.
However, the common assumption is that the number of content-types
does not scale with the number of peers, and that a request can be
served in parallel by multiple servers (and with increased network
bandwidth as the number of peers with a specific content-type
increases) which is fundamentally different from our setting.

Finally, our work is also related to the vast literature on content
replacement algorithms in server/web cache management. As discussed in
\cite{Wang99}, parameters of the content (e.g., how large is the
content, when was it last requested) are used to derive a cost, which
in-turn, is used to replace content.
Examples of
algorithms that have a cost-based interpretation include the Least
Recently Used (LRU) policy, the Least Frequently Used (LFU) policy,
and the Max-Size policy \cite{Size96}. 
We refer to \cite{Wang99} for a survey of web caching schemes. There
is a huge amount of work on the performance of replication strategies
in single-cache systems; however the analysis of adaptive caching
schemes in distributed cache systems under stochastic models of
arrivals and departures is very limited.


\section{Conclusions}\label{conclusion}
In this paper, we considered the high dimensional setting where the
number of servers, the number of content-types, and the number of
requests to be served over any time interval all scale as $\OO(n)$;
further the demand statistics are not known a-priori. This setting is
motivated by the enormity of the contents and their time-varying
popularity which prevent the consistent estimation of demands.

The main message of this paper is that in such settings, separating
the estimation of demands and the subsequent use of the estimations to
design optimal content placement policies (``learn-and-optimize''
approach) is order-wise suboptimal. This is in contrast to the
low dimensional setting, where the existence of a constant bound on the
number of content-types allows asymptotic optimality of a
learn-and-optimize approach.


\bibliographystyle{unsrt}
\bibliography{myref2}



\end{document}